\documentclass{article}
\usepackage[utf8]{inputenc}
\usepackage{float}
\usepackage{jheppub}
\usepackage{amsmath}
\usepackage[english]{babel}
\usepackage{amssymb}
\usepackage{mathtools}
\usepackage{color}
\definecolor{ao}{rgb}{0.0, 0.5, 0.0}

\usepackage[caption = false]{subfig}

\usepackage{graphicx}

\usepackage{tikz}
\usepackage{comment}
\usepackage{capt-of}
\usepackage{amsthm}
\usepackage{hyperref}
\usepackage{aesupp}
\usepackage{lmodern}

\hypersetup{
    colorlinks=true,
    linkcolor=blue
    }

\usepackage{bm}
\usepackage{ftnxtra}
\usepackage{fnpos}

\newtheorem{remark}{Remark}
\newtheorem{theorem}{Theorem}
\newtheorem{proposition}{Proposition}
\newtheorem{corollary}{Corollary}
\newtheorem{definition}{Definition}
\newtheorem{lemma}{Lemma}

\def\0{\mbox{\tiny $0$}}
\def\1{\mbox{\tiny $1$}}
\def\2{\mbox{\tiny $2$}}
\def\3{\mbox{\tiny $3$}}
\def\4{\mbox{\tiny $4$}}
\def\5{\mbox{\tiny $5$}}
\def\6{\mbox{\tiny $6$}}
\def\7{\mbox{\tiny $7$}}
\def\8{\mbox{\tiny $8$}}
\def\9{\mbox{\tiny $9$}}
\def\rhoD{\tilde\rho}

\def\r{\rho_{_{\mu}}}
\def\r{\rangle}

\def\l{\langle}
\def\m{\bar{m}}

\def\n{\bar{n}}

\def\b{\beta^{'}}

\newcommand{\SOMMA}[2]{\displaystyle\sum\limits_{#1}^{#2}}

\long\def \beq#1\eeq {\begin{equation} #1 \end{equation}}
\long\def \beaq#1\eeaq {\begin{equation}\begin{aligned} #1 \end{aligned}\end{equation}}
\long\def \bes#1\ees {\begin{equation}\begin{split} #1 \end{split} \end{equation}}
\long\def \bea#1\eea {\begin{eqnarray} #1 \end{eqnarray}}
\long\def \bse[#1]#2\ese {\begin{subequations}\label{#1}\begin{align} #2 \end{align}\end{subequations}}

\newcommand{\si}{\sigma_i}

\title{Parallel Learning by  Multitasking Neural Networks}

\author[a]{Elena Agliari,}
\author[b,e]{Andrea Alessandrelli, }
\author[c,e]{Adriano Barra, }
\author[d,e,f]{Federico Ricci-Tersenghi.}

\affiliation[a]{Dipartimento di Matematica, Sapienza Universit\`a di Roma, Piazzale Aldo Moro, 5, 00185, Roma, Italy.}

\affiliation[b]{Dipartimento di Informatica, Università di Pisa, Lungarno Antonio Pacinotti, 43, 56126, Pisa Italy.}

\affiliation[c]{Dipartimento di Matematica e Fisica,  Universit\`a  del Salento, Via per Arnesano, 73100, Lecce, Italy.}

\affiliation[d]{Dipartimento di Fisica, Sapienza Universit\`a di Roma, Piazzale Aldo Moro 2, I-00185 Roma, Italy.}

\affiliation[e]{Istituto Nazionale di Fisica Nucleare, Sezioni di Roma1 e Lecce, Italy.}

\affiliation[f]{CNR-Nanotec, Rome unit, 00185 Rome, Italy.}

\abstract{A modern challenge  of Artificial Intelligence is learning  multiple patterns at once (i.e. {\em parallel learning}). While this can not be accomplished by standard Hebbian associative neural networks, in this paper we show  how the Multitasking Hebbian Network (a variation on theme of the Hopfield model working on sparse data-sets) is naturally able to perform this complex task. We focus on systems processing in parallel a finite (up to logarithmic growth in the size of the network) amount of patterns, mirroring the low-storage level of standard associative neural networks at work with pattern recognition. For mild dilution in the patterns, the network handles them  hierarchically, distributing the  amplitudes of their signals as power-laws w.r.t. their information content  (hierarchical regime), while, for strong dilution, all the signals pertaining to all the patterns are raised with the same strength (parallel regime). 
\newline
Further, confined to the low-storage setting (i.e., far from the spin glass limit), the presence of a teacher neither alters the multitasking performances nor changes the thresholds for learning: the latter are the same whatever the training protocol is supervised or unsupervised.  Results obtained through statistical mechanics, signal-to-noise technique and Monte Carlo simulations are overall in perfect agreement and carry interesting  insights on {\em multiple learning at once}: for instance, whenever the cost-function of the model is minimized {\em in parallel on several patterns} (in its description via Statistical Mechanics), the same happens to the standard sum-squared error Loss function (typically used in Machine Learning). 
}

\begin{document}

\maketitle

\section{Introduction}\label{LaPrimina}
Typically, Artificial Intelligence has to deal with several inputs occurring at the same time: for instance, think about automatic driving, where it has to distinguish and react to different objects (e.g., pedestrians, traffic lights, riders, crosswalks) that may appear simultaneously. Likewise, when a biological neural network learns, it is rare that it has to deal with one single input per time\footnote{It is enough to note that, should serial learning take place rather than parallel learning, Pavlov's Classical Conditioning would not be possible \cite{Aquaro-Pavlov}.}: for instance, while trained in the school to learn any single letter, we are also learning about the composition of our alphabets. In this perspective, when stating that neural networks operate {\em in parallel}, some caution in potential  ambiguity should be paid. To fix ideas, let us focus on the Hopfield model \cite{Hopfield}, the {\em harmonic oscillator} of associative neural networks accomplishing pattern recognition \cite{Amit,Coolen}: its neurons indeed operate synergistically in parallel but with the purpose of retrieving one single pattern per time, not several simultaneously \cite{Amit,AGS2,Taro}. A parallel processing where multiple patterns are simultaneously retrieved cannot be accessible to the standard Hopfield networks as long as each pattern is fully informative, namely its vectorial binary representation is devoid of blank entries. On the other hand, when a fraction of entries can be blank \cite{Aquaro-Pavlov} multiple-pattern retrieval is potentially achievable by the network. Intuitively,  this can be explained by noticing that the overall number of neurons making up the networks -- and thus available for information processing -- equals the length of the binary vectors codifying the patterns to be retrieved, hence, as long as these vectors contain information in all their entries, there is no free room for dealing with multiple patterns.  
Conversely, the multitasking neural networks, introduced in \cite{Barra-PRL2012}, are able to overcome this limitation and have been shown to succeed in retrieving multiple patterns simultaneously, just by leveraging the presence of lacun\ae\ in the patterns stored by the network. The emerging pattern-recognition properties have been extensively investigated at medium storage (i.e., on random graphs above the percolation threshold) \cite{Coolen-Medium}, at high storage (i.e., on random graphs below the percolation threshold) \cite{Coolen-High} as well as on scale-free \cite{Barra-PRL2014} and hierarchical \cite{Barra-PRL2015}  topologies. 

However, while the study of the parallel retrieval capabilities of these multi-tasking networks is nowadays over, the comprehension of their parallel learning capabilities just started and it is the main focus of the present paper. 

In these regards it is important to stress that the Hebbian prescription  has been recently revised to turn it from a storing rule (built on a set of already definite patterns, as in the original Amit-Gutfreund-Sompolinksy (AGS) theory) into a genuine learning rule (where unknown patterns have to be inferred by experiencing solely  a sample of their corrupted copies), see e.g., \cite{Agliari-Emergence,Aquaro-EPL,Fede-Data2016}\footnote{While Statistical Learning theories appeared in the Literature a long time ago, see e.g. \cite{ContrDiv,SompoLearning,Angel} for the original works and \cite{Fede-Data2012,Fede-Data2014,Carleo, Lenka} for updated references, however the statistical mechanics of Hebbian learning was not deepened in these studies.}. 
\newline
In this work we merge these extensions of the bare AGS theory and use definite patterns (equipped with blank entries) to generate a sparse data-set of corrupted examples, that is the solely information experienced by the network: we aim to highlight the role of lacun\ae\ density and of the data-set size and quality on the network performance, in particular deepening the way the network learns simultaneously the patterns hidden behind the supplied examples.  
In this investigation we focus on the low-storage scenario (where the number of definite patterns grows sub-linearly with the volume of the network) addressing both the {\em supervised} and the {\em unsupervised} setting.

The paper is structured as follows:  the main text has three Sections. Beyond this Introduction provided in Section \ref{LaPrimina}, in Section \ref{SezioneDue} we revise the multi-tasking associative network; once briefly summarized its parallel retrieval capabilities (Sec. \ref{Sezione2.1}),  we introduce a simple data-set the network has to cope with in order to move from the simpler storing of patterns to their learning from examples (Sec. \ref{Gramsci}). Next, in Section \ref{sec:RS_Guerra_sup} we provide an exhaustive statistical mechanical picture of the network's emergent information processing capabilities by taking advantage of Guerra's interpolation techniques: in particular, focusing on the Cost function (Sec. \ref{SezioneTreUno}), we face the {\em big-data} limit (Sec. \ref{sec:datalimit_sup}) and we deepen the nature of the phase transition the network undergoes as ergodicity breaking spontaneously takes place (Sec. \ref{cazzopeppe}). Sec. \ref{cazzarola} is entirely dedicated to provide phase diagrams (namely plots in the space of the control parameters where different regions depict different global computational capabilities). Further, before reaching conclusions and outlooks as reported in Sec. \ref{SezioneConclusiva}, in Sec. \ref{SezioneTreTre} we show how the network's Cost function (typically used in Statistical Mechanics) can be sharply related to  standard Loss functions (typically used in Machine Learning) to appreciate how parallel learning effectively lowers several Loss functions at once.
\newline
In the Appendices we fix a number of subtleties: in Appendix \ref{app:sampling} we provide a more general setting for the sparse data-sets considered in this research\footnote{In the main text we face the simplest kind of pattern's dilution, namely we just force to be blank the same fraction of their entries whose position is preserved in the generation of the data-sets (hence whenever the pattern has a zero, in all the examples it gives rise to, the zero will be kept), while in the appendix we relax this assumption (and blank entries can move along the examples yet preserving their amount). As in the thermodynamic limit the theory is robust w.r.t. these structural details we present as a main theme the simplest setting  and in the appendix \ref{app:sampling} the more cumbersome one.}, while in Appendix \ref{AppEntropy} we inspect the relative entropies of these data-sets and finally in Appendix \ref{AppendiceC} we provide a revised version of the Signal-to-Noise technique (that allows to evaluate computational shortcuts beyond providing an alternative route to obtain the phase diagrams).  Appendices \ref{app:calc} and \ref{app:proof} give details on  calculations, plots and proofs of the main theorems.

 \section{Parallel learning in multitasking Hebbian neural networks}\label{SezioneDue}
\subsection{A preliminary glance at the emergent parallel retrieval capabilities}\label{Sezione2.1}
Hereafter, for the sake of completeness, we briefly review the retrieval properties of the multitasking Hebbian network in the low-storage regime, while we refer to \cite{Barra-PRL2012,ElenaMult1} for an extensive treatment. 

\begin{definition} \label{def:Hop}
Given $N$ Ising neurons $\sigma_i = \pm 1$ ($i=1,...,N$), and $K$ random patterns $\boldsymbol \xi^{\mu}$ ($\mu=1,...,K$), each of length $N$, whose entries are i.i.d. from
\begin{equation}\label{eq:quenched_xi}
\mathbb  P(\xi_i^{\mu}) = \frac{(1-d)}{2}\delta_{\xi_i^{\mu},-1} +  \frac{(1-d)}{2}\delta_{\xi_i^{\mu},+1} + d  \delta_{\xi_i^{\mu},0},   
\end{equation}
where $\delta_{i,j}$ is the Kronecker delta and $d \in [0,1]$, the Hamiltonian (or cost function) of the system reads as   
\begin{equation} \label{eq:cost_0}
\mathcal H_N(\boldsymbol \sigma|\boldsymbol \xi) :=-\frac{1}{2N} \sum_{\substack{i,j \\ i\neq j}}^{N,N}\left(\sum_{\mu=1}^K \xi_i^{\mu} \xi_j^{\mu}\right) \sigma_i\sigma_j.
\end{equation}
\end{definition}
The parameter $d$ tunes the ``dilution'' in pattern entries: if $d=0$ the standard Rademacher setting of AGS theory is recovered, while for $d = 1$ no information is retained in these patterns: otherwise stated, these vectors display, on average, a fraction $d$ of blank entries.
\begin{definition}
In order to assess the network retrieval performance we introduce the $K$ Mattis magnetizations 
\begin{equation}
m_{\mu} := \frac{1}{N} \sum_i^N \xi_i^{\mu}\sigma_i, ~ \mu=1,...,K,
\end{equation}
which quantify the overlap between the generic neural configuration $\boldsymbol \sigma$ and the $\mu^{th}$ pattern.
\end{definition} 
Note that the cost function \eqref{eq:cost_0} can be recast as a quadratic form in $m_{\mu}$, namely 
\begin{equation}\label{abucione}
\mathcal H_N(\boldsymbol \sigma|\boldsymbol \xi)=
-\frac{N}{2}\sum_{\mu} m_{\mu}^2 + \frac{K}{2},
\end{equation}
where the term $K/2$ in the r.h.s. stems from diagonal terms ($i=j$) in the sum at the r.h.s. of eq. \ref{eq:cost_0} and in the low-load scenario (i.e., $K$ grows sub-linearly with $N$) can be neglected in the thermodynamic limit ($N\to\infty$).
\newline
As we are going to explain, the dilution ruled by $d$ is pivotal for the network in order to perform parallel processing.
It is instructive to first consider a toy model handling just $K=2$ patterns: let us assume, for simplicity, that the first pattern $\boldsymbol \xi^1$ contains information (i.e., no blank entries) solely in the first half of its entries and the second pattern $\boldsymbol \xi^2$ contains information solely in the second half of its entries, that is
\begin{equation}
\boldsymbol \xi^1=(\underbrace{\xi_1^1, ..., \xi_{N/2}^1}_{\in \{-1, +1\}^{\frac{N}{2}}},\underbrace{0,..., 0}_{\in \{0\}^{\frac{N}{2}}}), ~~~  \boldsymbol \xi^2=(\underbrace{0, ..., 0}_{\in \{0\}^{\frac{N}{2}}},\underbrace{\xi_{N/2+1}^1,..., \xi^1_N}_{\in \{-1, +1\}^{\frac{N}{2}}})
\end{equation}
%
%
%
Unlike the standard Hopfield reference ($d=0$), where the retrieval of one pattern employs all the resources and there is no chance to retrieve any other pattern, not even partially (i.e., as  $m_1 \to 1$ then $m_2 \approx 0$ because patterns are orthogonal for large $N$ values in the standard random setting), here nor $m_1$ neither $m_2$ can reach the value $1$ and therefore the complete retrieval of one of the two still leaves resources for the retrieval of the other. In this particular case, the minimization of the cost function $\mathcal H_N(\boldsymbol \sigma|\boldsymbol \xi)= -\frac{N}{2}\left(m_1^2 + m_2^2\right)$ is optimal when \emph{both} the magnetizations are equal to one-half, that is when they both saturate their upper bound.
In general, for arbitrary dilution level $d$, the minimization of the cost function requires the network to be in one of the following regimes
\begin{itemize}
    \item {\em hierarchical scenario}: for values of dilution not too high (i.e., $d<d_c$, {\em vide infra}), one of the two patterns is fully retrieved (say $m_1 \approx 1-d$) and the other is retrieved to the largest extent given the available resources, these being constituted by, approximately, the $Nd$ neurons corresponding to the blank entries in $\bm \xi^1$  (thus, $m_2 \approx d(1-d)$), and so on if further patterns are considered.
\item {\em parallel scenario}: for large values of dilution (i.e., above a critical threshold $d_c$), the magnetizations related to all the patterns raise and the signals they convey share the same amplitude.
\end{itemize}
\begin{figure}
\begin{center}
\includegraphics[width=15cm]{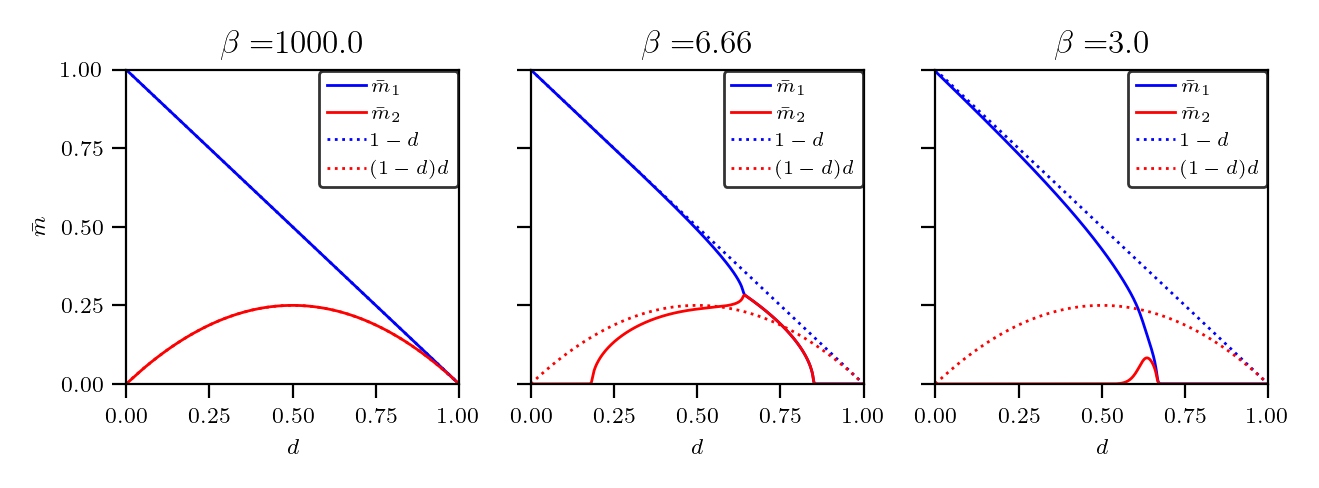} 
\caption{Numerical solutions of the two self-consistent equations \eqref{ammerda1} and \eqref{ammerda2} obtained for $K=2$, see \cite{Barra-PRL2012}, as a function of $d$ and for different choices of $\beta$: in the $d \to 0$ limit  the Hopfield serial retrieval is recovered (one magnetization with intensity one and the other locked at zero), for $d \to 1$ the network ends up in the parallel regime (where all the magnetizations acquire the same value), while for intermediate values of dilution the hierarchical ordering prevails (both the magnetizations are raised, but their amplitude is different). \label{fig:UnoTaccitua}}
\end{center}
\end{figure}

In general, in this type of neural network, the {\em pure state ansatz}\footnote{In this state the neurons are aligned with one of the patterns and, without loss of generality, here we refer to $\mu=1$.} $\bm m = (1,0, 0, ...,0)$, that is $\sigma_i = \xi^{1}_i$ for $i=1,...,N$, barely works and parallel retrieval is  often favored. In fact, for $K \geq 2$, at relatively low values of pattern dilution $d_1$ and in the zero-noise limit $\beta \to \infty$, one can prove the validity of the so-called \emph{hierarchical ansatz} \cite{Barra-PRL2012} as we briefly discuss: one pattern, say $\boldsymbol \xi^1$, is perfectly retrieved and displays a Mattis magnetization $m^1 \approx (1-d)$; a fraction $d$ of neurons is not involved and is therefore available for further retrieval, with any remaining pattern, say $\boldsymbol \xi^2$, which yields $m_2 \sim (1 - d)d$; proceeding iteratively, one finds $m_{\ell} = d^{\ell-1}(1-d)$ for $\ell =1, ..., \hat K$ and the overall number $\hat K$ of patterns simultaneously retrieved corresponds to the employment of all the resources. Specifically, $\hat K$ can be estimated by setting $\sum_{\ell=0}^{\hat K-1} (1-d) d^{\ell} =1$, with the cut-off at finite $N$ as $(1-d)d^{\hat K-1} \geq N^{-1}$, due to discreteness: for any fixed and finite $d$, this implies $K \lesssim \log N$, which can be thought of as a ``parallel low-storage'' regime of neural networks. It is worth stressing that, in the above mentioned regime of low dilution, the configuration leading to $m_{\ell} = d^{\ell-1}(1-d)$ for $\ell =1, ..., \hat K$ is the one which minimizes the cost function. 
\begin{figure}
\begin{center}
\includegraphics[width=12cm, height= 4.0cm]{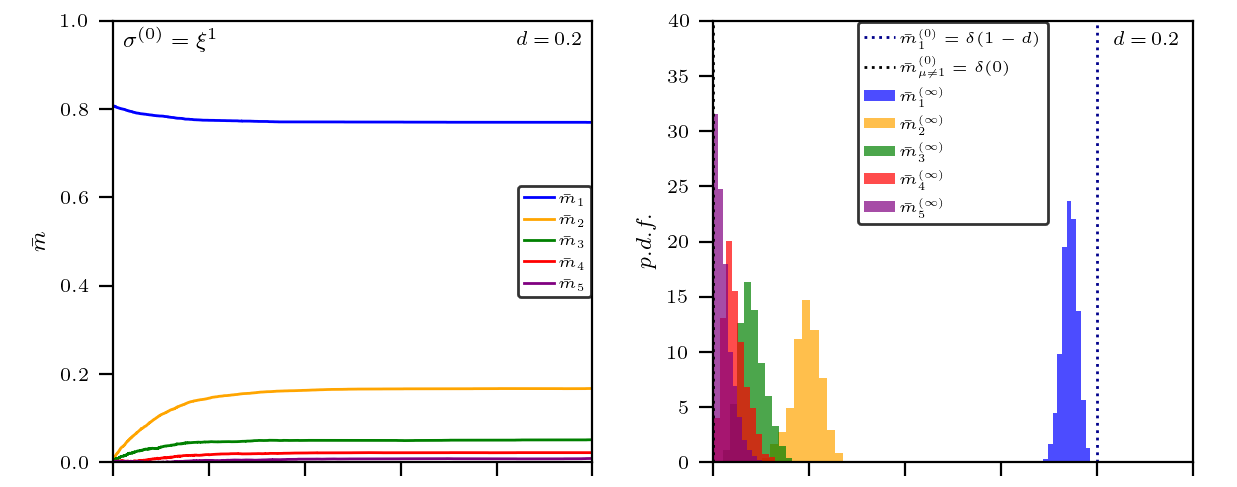} 
\includegraphics[width=12cm, height= 4.0cm]{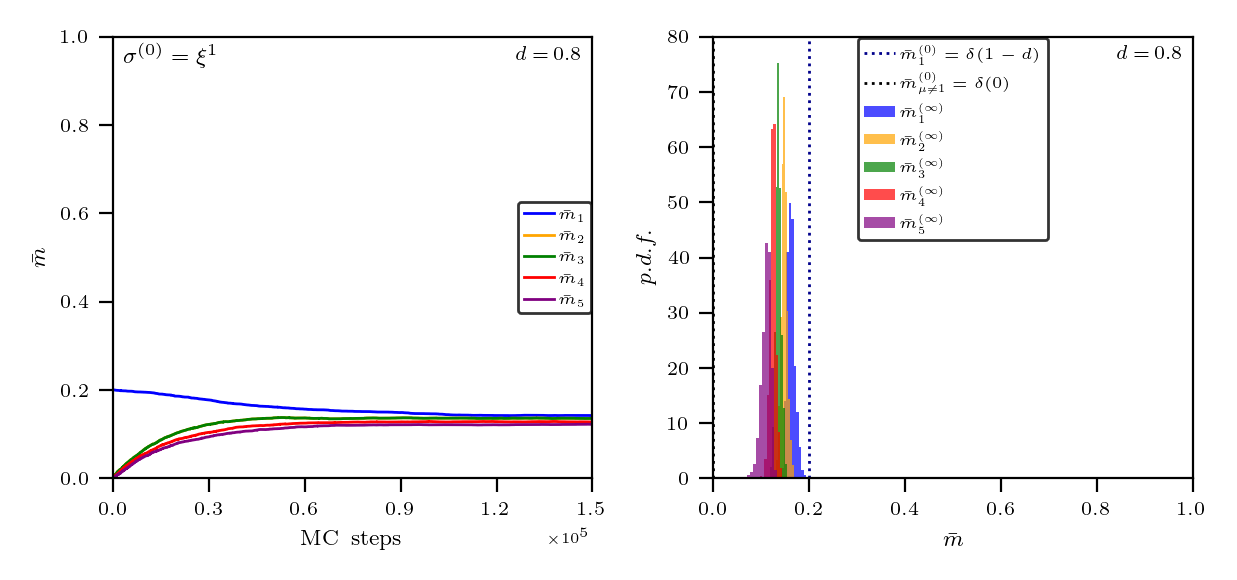} 
\caption{We report two examples of Monte Carlo dynamics until  thermalization within the hierarchical (upper plots, dilution level $d=0.2$) and parallel (lower plots, dilution level $d=0.8$) scenarios respectively. These plots confirm that the picture provided by statistical mechanics is actually dynamically reached by the network. We initialize the network sharply in a pattern as a Cauchy condition (represented as the dotted blue Dirac delta peaked at the pattern in the second columns) and, in the first column,  we show the stationary values of the various Mattis magnetizations pertaining to  different patterns, while in the second column we report their histograms achieved by sampling 1000 independent Monte Carlo simulations: starting from a sequential retrieval regime, the network ends up in a multiple retrieval mode, hierarchical vs parallel depending on the level of dilution in the patterns. \label{fig:DueTaccidete}}
\end{center}
\end{figure}
The hierarchical retrieval state 
$\bm m = (1-d)\left(1,d,d^2,d^3,\cdots\right)$ can also be specified in terms of neural configuration as  \cite{Barra-PRL2012}
%
%
\begin{equation}
    \sigma_i^* = \xi_i^1+\SOMMA{\nu=2}{\hat K}\xi_i^\nu\prod\limits_{\rho=1}^{\nu-1}\delta_{\xi_i^{\rho},0}\,.
    \label{eq:parallel_ans}
\end{equation}
This organization is stable until a critical dilution level $d_c$ is reached where $m_1 \sim \sum_{k>1}m_k$ \cite{Barra-PRL2012}, beyond that level the network undergoes a rearrangement and a new organization called {\em parallel ansatz} supplants the previous one. Indeed for high values of dilution (i.e $d \to 1$) it is immediate to check that the ratio among the various intensities of all the magnetizations stabilizes to the value one, i.e. $(m_{k}/m_{k-1}) \sim d^{k-1}(1-d)/d^{k-2}(1-d) \to 1$, hence in this regime all the magnetizations are raised with the same strength and the network is operationally set in a fully parallel retrieval mode:  the parallel retrieval state simply reads $\bm m = (\bar{m})\left(1,1,1,1,\cdots\right)$.
This picture is confirmed  by the plots shown in  Fig. \ref{fig:UnoTaccitua} and obtained by solving the self-consistency equations for the Mattis magnetizations related to the multitasking Hebbian network equipped with $K=2$ patterns that read as \cite{Barra-PRL2012}
\begin{eqnarray}\label{ammerda1}
m_1 &=& d(1-d) \tanh(\beta m_1) + \frac{(1-d)^2}{2}\left\{\tanh[\beta (m_1+m_2)] + \tanh [\beta (m_1-m_2)]\right\},\\  \label{ammerda2}
m_2 &=& d(1-d) \tanh(\beta m_1) + \frac{(1-d)^2}{2}\left\{\tanh[\beta (m_1+m_2)] - \tanh [\beta (m_1-m_2)]\right\}
\end{eqnarray}
where $\beta \in \mathbb R^+$ denotes the level of noise. 
\newline
We remark that these hierarchical or parallel organizations of the retrieval, beyond emerging naturally within the equilibrium description provided by Statistical Mechanics, are actually the real stationary states of the dynamics of these networks at work with diluted patterns as shown in Figure \ref{fig:DueTaccidete}.

\subsection{From parallel storing to parallel learning}\label{Gramsci}
In this section we revise the multitasking Hebbian network \cite{Barra-PRL2012,ElenaMult1} in such a way that it can undergo a \emph{learning} process instead of a simple \emph{storing} of patterns. In fact, in the typical learning setting, the set of definite patterns, hereafter promoted to play as ``archetypes'', to be reconstructed by the network is not available, rather, the network is exposed to examples, namely noisy versions of these archetypes. 

As long as enough examples are provided to the network, this is expected to correctly form its own representation of the archetypes such that, in further expositions to a new example related to a certain archetype, it will be able to retrieve it and, since then, suitably generalize it. 

This generalized Hebbian kernel has recently been introduced to encode unsupervised \cite{Agliari-Emergence} and supervised \cite{Aquaro-EPL} learning processes and, in the present paper, these learning rules are modified in order to deal with diluted patterns.

\bigskip

First, let us define the data-set these networks have to cope with: the archetypes are randomly drawn from the distribution \eqref{eq:quenched_xi}. 
Each archetype $\boldsymbol \xi^{\mu}$ is then used to generate a set of $M_{\mu}$ perturbed versions, denoted as $\boldsymbol{\eta}^{\mu,a}$ with $a=1,...,M_{\mu}$ and $\boldsymbol{\eta}^{\mu,a} \in \{ -1,0, +1\}^N$.
Thus, the overall set of examples to be supplied to the network is given by $\bm \eta =  \{\boldsymbol{\eta}^{\mu,a}\}_{\mu=1,...,K}^{a=1,...,M_{\mu}}$. Of course, different ways to sample examples are conceivable: for instance, one can require that the position of blank entries appearing in $\bm \xi^{\mu}$ is preserved over all the examples $\{\bm \eta ^{\mu,a}\}_{a=1,...,M_{\mu}}$, or one can require that only the number of blank entries $\sum_{i=1}^N \delta_{\xi_i^{\mu},0}$ is preserved (either strictly or in the average). Here we face the first case because it requires a simpler notation, but we refer to Appendix \ref{app:sampling} for a more general treatment.  
%
%
\begin{definition}
The entries of each examples are depicted following 
\begin{equation}
\label{eq:quenched_chi}
\mathbb{P}(\eta_i^{\mu,a}|\xi_{i}^{\mu})= \frac{1+r_\mu}{2}\delta_{\eta_i^{\mu,a},\xi_i^\mu} +  \frac{1-r_\mu}{2}\delta_{\eta_i^{\mu,a},-\xi_i^\mu}, 
\end{equation}
for $i=1, \hdots, N$ and $\mu=1,\hdots, K$. Notice that $r_\mu$ tunes the data-set quality: as $r_\mu \to 1$ examples belonging to the $\mu$-th set collapse on the archetype $\boldsymbol \xi^{\mu}$, while as $r_\mu \to 0$ examples turn out to be uncorrelated with the related archetype $\boldsymbol \xi^{\mu}$.
\end{definition}
As we will show in the next sections, the behavior of the system depends on the parameters $M_{\mu}$ and $r_{\mu}$ only through the combination $\frac{1-r_\mu^2}{M_\mu r_\mu^2}$, therefore, as long as the ratio $\frac{1-r_\mu^2}{M_\mu r_\mu^2}$ is $\mu$-independent, the theory shall not be affected by the specific choice of the archetype. 
Thus, for the sake of simplicity, hereafter we will consider $r$ and $M$ independent of $\mu$ and we will pose $\rho:=\frac{1-r^2}{M r^2}$.
Remarkably, $\rho$ plays as an information-content control parameter \cite{Aquaro-EPL}: to see this, let us focus on
the $\mu$-th pattern and $i$-th digit, whose related block is $ \boldsymbol{\eta}_i^{\mu}=(\eta_i^{\mu,1},\eta_i^{\mu,2},\hdots,\eta_i^{\mu,M})$, the error probability for any
single entry is $\mathcal{P}(\xi^{\mu}_i \neq 0)\mathcal{P}(\eta^{\mu,a}_i \neq \xi_i^\mu) = (1-d)(1 - r_\mu)/2$ and, by applying the majority rule on the block, we get $\mathcal{P}(\xi^{\mu}_i \neq 0)\mathcal P(\mathrm{sign}(\sum\limits_a \eta^{\mu,a}_i)\xi_i^\mu = -1) \underset{M\gg 1}{\approx}\frac{(1-d)}{2}\left[1 -\mathrm{erf}\left(1/\sqrt{2\rho}\right)\right]$  thus, by computing the conditional entropy $H_d(\xi_i^\mu|\bm \eta_i^\mu)$, that quantifies the amount
of information needed to describe the original message $\xi_i^\mu$ given the related block $\bm\eta_i^\mu$, we get
\begin{equation}
\begin{array}{lll}
     H_d(\xi_i^\mu|\bm \eta_i^\mu) &=& -\left[\dfrac{1+d}{2} +\dfrac{1-d}{2}\mathrm{erf}\left(\dfrac{1}{\sqrt{2\rho}}\right)\right]\;\,\log \left[\dfrac{1+d}{2} +\dfrac{1-d}{2}\mathrm{erf}\left(\dfrac{1}{\sqrt{2\rho}}\right)\right]
     \\\\
     &&-\left[\dfrac{1-d}{2} -\dfrac{1-d}{2}\mathrm{erf}\left(\dfrac{1}{\sqrt{2\rho}}\right)\right]\;\,\log\left[\dfrac{1-d}{2} -\dfrac{1-d}{2}\mathrm{erf}\left(\dfrac{1}{\sqrt{2\rho}}\right)\right]
\end{array}
     \label{eq:entropy_standard}
 \end{equation}
which is  monotonically increasing with $\rho$. Therefore, with a slight abuse of language, in the following $\rho$ shall be referred to as {\em data-set entropy}.
\newline
The available information is allocated directly in the synaptic coupling  among neurons (as in the standard Hebbian storing), as specified by the following supervised and unsupervised generalization of the multitasking Hebbian network:

\begin{definition}
Given $N$ binary neurons $\sigma_i = \pm 1$, with $i \in (1,...,N)$, the cost function (or \textit{Hamiltonian}) of the multitasking Hebbian neural network in the supervised regime is 
\begin{equation}
   \mathcal{H}^{(sup)}_{N,K, d,M,r}(\boldsymbol{\sigma} \vert \bm \eta) = -\dfrac{1}{2N}\dfrac{1}{ (1-d)(1+\rho)}\SOMMA{\mu=1}{K}\SOMMA{i,j=1}{N,N}\left(\dfrac{1}{M r}\SOMMA{a=1}{M}\eta_i^{\mu,a}\right)\left(\dfrac{1}{M r}\SOMMA{b=1}{M}\eta_j^{\mu,b}\right)\sigma_i\sigma_j.
    \label{def:H_dil_sup}
\end{equation}
\end{definition}
\begin{definition}
Given $N$ binary neurons $\sigma_i = \pm 1$, with $i \in (1,...,N)$, the cost function (or \textit{Hamiltonian}) of the multitasking Hebbian neural network in the unsupervised regime is 
\begin{equation}
   \mathcal{H}^{(unsup)}_{N,K, d,M,r}(\boldsymbol{\sigma} \vert \bm \eta) = -\dfrac{1}{2N}\dfrac{1}{ (1-d)(1+\rho)}\SOMMA{\mu=1}{K}\SOMMA{i,j=1}{N,N} \left(\frac{1}{Mr^2} \SOMMA{a=1}{M}\eta_i^{\mu,a}\eta_j^{\mu,a} \right) \sigma_i\sigma_j.
    \label{def:H_dil_unsup2}
\end{equation}
\end{definition}

\begin{remark}
The factor $(1-d)(1+\rho)$ appearing in \eqref{def:H_dil_sup} corresponds to 
$\mathbb E_{\xi},\mathbb E_{(\eta|\xi)} \left[\sum\limits_a \eta_i^{\mu,a}/(Mr)\right]^2$ and it plays as a normalization factor. A similar factor is also inserted in \eqref{def:H_dil_unsup2}.
\end{remark}

\begin{remark}
By direct comparison between \eqref{def:H_dil_sup} and \eqref{def:H_dil_unsup2}, the role of the ``teacher'' in the supervised setting is evident: in the unsupervised scenario, the network has to handle all the available examples regardless of their archetype label, while in the supervised counterpart a teacher has previously grouped examples belonging to the same archetype together (whence the double sum on $a=(1,...,M)$ and on $b=(1,...,M)$ appearing in eq. (\ref{def:H_dil_sup}), that is missing in eq. (\ref{def:H_dil_unsup2})).
\end{remark}

We investigate the model within a canonical framework: we introduce the Boltzmann-Gibbs measure
\begin{equation}
\mathcal{P}^{(sup,unsup)}_{N,K, \beta, d,M,r}(\bm \sigma  \vert \bm \eta) := \frac{\exp [- \beta \mathcal{H}^{(sup,unsup)}_{N,K, d,M,r}(\bm \sigma  \vert \bm \eta)]}{ \mathcal{Z}^{(sup,unsup)}_{N,K, \beta, d,M,r} (\bm \eta)},
\end{equation}
where 
\begin{equation} \label{eq:ppff}
\mathcal{Z}^{(sup,unsup)}_{N,K, \beta, d,M,r} (\bm \eta) := \sum_{\bm \sigma} \exp \left[ -\beta \mathcal H^{(sup,unsup)}_{N,K, d,M,r}(\bm \sigma \vert \bm \eta)\right]
\end{equation}
is the normalization factor, also referred to as partition function,  
and the parameter $\beta \in \mathbb{R}^+$, rules the broadness of the distribution in such a way that for $\beta \to 0$ (infinite noise limit) all the $2^N$ neural configurations are equally likely, while for $\beta \to \infty$ the distribution is delta-peaked at the configurations corresponding to the minima of the Cost function.

The average performed over the Boltzmann-Gibbs measure is denoted as
\begin{equation} \label{omegaNKM}
\omega_{N,K, \beta, d,M,r}^{(sup, unsup)}[\cdot]:=  \sum_{\boldsymbol \sigma}^{2^N} ~ \cdot ~ \mathcal{P}^{(sup,unsup)}_{N,K, \beta, d,M,r}(\bm \sigma  \vert \bm \eta).
\end{equation}
Beyond this average, we shall also take the so-called \emph{quenched} average, that is the average over the realizations of archetypes and examples, namely over the distributions \eqref{eq:quenched_xi} and \eqref{eq:quenched_chi}, and this is denoted as
\begin{equation}
\mathbb{E}[\cdot]=\mathbb{E}_{\xi}\mathbb{E}_{(\eta|\xi)} [\cdot].
\end{equation}
     
\begin{definition}    
The quenched statistical pressure of the network at finite network size $N$ reads as 
\begin{align}
    \mathcal A_{N,K, \beta, d,M,r}^{(sup,unsup)} = \frac{1}{N}\mathbb{E}\log  \mathcal{Z}^{(sup,unsup)}_{N,K, \beta, d,M,r} (\bm \eta).
    \label{PressureDef_sup}
\end{align}
In the thermodynamic limit we pose
\begin{equation}
\mathcal A^{(sup,unsup)}_{K, \beta, d,M,r} = \lim_{N \to \infty} \mathcal A^{(sup,unsup)}_{N,K, \beta, d,M,r}.
\end{equation} 
We recall that the statistical pressure equals the free energy times $-\beta$ (hence they convey the same information content). 
\end{definition}


\begin{definition} 
The network capabilities can be quantified by introducing the following order parameters,   for $\mu=1, \hdots , K$,
\begin{equation}
    \begin{array}{lll}
    \label{eq:orderparameters}
         m_\mu& :=\dfrac{1}{N}\SOMMA{i=1}{N}\xi_i^{\mu}\sigma_i,
            \\
         n_{\mu,a}& :=\dfrac{1}{(1+\rho)r}\dfrac{1}{N}\SOMMA{i=1}{N}\eta_i^{\mu,a}\sigma_i,
             \\
         n_\mu& :=\dfrac{1}{M}\SOMMA{a=1}{M}n_{\mu,a}=\dfrac{1}{(1+\rho)r}\dfrac{1}{N M}\SOMMA{i,a=1}{N,M}\eta_i^{\mu,a}\sigma_i, \\
    \end{array}
\end{equation}
\end{definition}
We stress that, beyond the fairly standard $K$ Mattis magnetizations $m_{\mu}$, which assess the alignment of the neural configuration $\boldsymbol \sigma$ with the archetype $\boldsymbol \xi^{\mu}$, we need to introduce also $K$ empirical Mattis magnetizations $n_{\mu}$, which compare  the alignment of the neural configuration with the average of the examples labelled with $\mu$, as well as $K \times M$ single-example Mattis magnetizations $n_{\mu,a}$, which measure the proximity between the neural configuration and a specific example.
An intuitive way to see the suitability of the $n_{\mu}$'s and of the $n_{\mu,a}$'s is by noticing that the cost functions $\mathcal H^{(sup)}$ and $\mathcal H^{(unsup)}$ can be written as a quadratic form in, respectively, $n_{\mu}$ and $n_{\mu,a}$; on the other hand, the $m_{\mu}$'s do not appear therein explicitly as the archetypes are unknowns in principle. 
\newline
Finally, notice that no spin-glass order parameters is needed here  (since we are working in the low-storage regime \cite{Amit,Coolen}).

\section{Parallel Learning: the picture by statistical mechanics}\label{sec:RS_Guerra_sup}
\subsection{Study of the Cost function and its related Statistical Pressure}\label{SezioneTreUno}
To inspect the emergent capabilities of these networks, we need to estimate the order parameters introduced in Equations \eqref{eq:orderparameters} and analyze their behavior versus the control parameters $K, \beta, d,M, r$. To this task we need an explicit expression of the statistical pressure in terms of these order parameters so to extremize the former over the latter. In this Section we carry on this investigation in the thermodynamic limit and in the low storage scenario by relying upon Guerra's interpolating techniques (see e.g., \cite{Barbier,Fachechi1,Guerra1,Guerra2}): the underlying idea is to introduce an interpolating statistical pressure whose extrema are the original model (which is the target of our investigation but we can be not able to address it directly)  and a simple one (which is usually a one-body model that we can solve exactly). We then start by evaluating the solution of the latter and next we propagate the obtained solution back to the original model by the fundamental theorem of calculus, integrating on the interpolating variable. Usually, in this last passage, one assumes replica symmetry, namely that the order-parameter fluctuations are negligible in the thermodynamic limit as this makes the integral propagating the solution analytical. In the low-load scenario replica symmetry holds exactly, making the following calculation rigorous. 
In fact, as long as $K/N \to 0$ while $N \to \infty$, the order parameters self-average around their means \cite{Bovier,Tala}, that will be denoted by a bar, that is
\begin{eqnarray}\label{PdiEmme}
\lim_{N \to \infty} \mathcal P_{N,K, \beta, d,M,r}(m_{\mu}) &=& \delta\left(m_{\mu}-\bar{m}_{\mu} \right),\ \ \ \forall \mu \in (1,...,K), \\  \label{PdiEmme}  
\lim_{N \to \infty} \mathcal P_{N,K, \beta, d,M,r}(n_{\mu}) &=& \delta\left(n_{\mu}-\bar{n}_{\mu} \right),\ \ \ \ \  \forall \mu \in (1,...,K),
\end{eqnarray}
where $\mathcal P_{N,K, \beta, d,M,r}$ denotes the  Boltzmann-Gibbs probability distribution for the observables considered. We anticipate that the mean values of these distributions are independent of the training (either supervised or unsupervised) underlying the Hebbian kernel.

Before proceeding, we slightly revise the partition functions \eqref{eq:ppff} by inserting an extra term in their exponents because it allows to apply the functional generator technique to evaluate the Mattis magnetizations. This implies the following modification, respectively in the supervised and unsupervised settings, of the partition function 
\begin{definition}
Given the interpolating parameter $t \in [0,1]$, the auxiliary field $J$ and the constants $\{\psi_\mu\}_{\mu=1,\hdots,K} \in \mathbb{R}$  to be set a posteriori, Guerra's interpolating partition function for the supervised and unsupervised multitasking Hebbian networks is given, respectively, by
\begin{equation}\small
\begin{array}{lll}
     \mathcal{Z}^{(sup)}_{N,K, \beta, d,M,r}( \bm \eta ;J, t)=&\SOMMA{\lbrace\boldsymbol{\sigma}\rbrace}{} \displaystyle\int\,d\mu(z_{\mu}) \exp{\Bigg[J \sum_{\mu,i}  \xi_i^\mu \si+\dfrac{t\beta N (1+\rho)}{2(1-d)}\sum_\mu n_\mu^2(\boldsymbol{\sigma})}+(1-t)\dfrac{N}{2}\sum_\mu\psi_\mu \, n_\mu(\boldsymbol{\sigma})\Bigg].
\end{array}
\label{def:partfunct_GuerraRS}
\end{equation}
\normalsize
\begin{equation}\small
\begin{array}{lll}
\mathcal{Z}^{(unsup)}_{N,K, \beta, d,M,r}( \bm \eta ;J, t)=\SOMMA{\lbrace\boldsymbol{\sigma}\rbrace}{} \displaystyle\int\,d\mu(z_{\mu}) \exp{\Bigg[J \sum_{\mu,i} \xi_i^\mu \si+ \dfrac{t\beta N (1+\rho)}{2 (1-d)M}\SOMMA{\mu=1}{K}\SOMMA{a=1}{M} n_{\mu,a}^2(\boldsymbol{\sigma})+(1-t)N\sum_{\mu,a}\psi_{\mu} \, n_{\mu,a}(\boldsymbol{\sigma})\Bigg]}.
\end{array}
\label{def:partfunct_GuerraRS_unsup}
\end{equation}
\normalsize
\end{definition}
More precisely, we added the term $J \sum_\mu\sum_i \xi_i^\mu \si$ that allows us to to ``generate'' the expectation of the Mattis magnetization $m_\mu$ by evaluating the derivative w.r.t. $J$ of the quenched statistical pressure at $J=0$.
This operation is not necessary for {\em Hebbian storing}, where the Mattis magnetization is  a natural order parameter (the Hopfield Hamiltonian can be written as a quadratic form in $m_{\mu}$, as standard in AGS theory \cite{Amit}), while for {\em Hebbian learning}  (whose cost function can be written as a quadratic form in $n_{\mu}$, not in $m_{\mu}$ as the network does not experience directly the archetypes) we need such a term for otherwise the expectation of the Mattis magnetization would not be accessible. This operation gets redundant in the $M \to \infty$ limit, where $m_{\mu}$ and $n_{\mu}$ become proportional by a standard Central Limit Theorem (CLT) argument (see also Sec.~\ref{sec:datalimit_sup} and \cite{Aquaro-EPL}).
Clearly, $\mathcal{Z}^{(sup,unsup)}_{N,K, \beta, d,M,r} (\bm \eta) = \lim_{J \rightarrow 0} \mathcal{Z}^{(sup,unsup)}_{N,K, \beta, d,M,r}( \bm \eta ;J)$ and these generalized interpolating partition functions, provided in eq.s (\ref{def:partfunct_GuerraRS}) and (\ref{def:partfunct_GuerraRS_unsup}) respectively, recover the original models when $t=1$, while they return a simple one-body model  at $t=0$.
\newline
Conversely, the role of the $\psi_\mu$'s is instead that of mimicking, as close as possible, the true post-synaptic field perceived by the neurons.
%
%

These partition functions can be used to define a generalized measure and a generalized Boltzmann-Gibbs average that we indicate by $\omega^{(sup,unsup)}_t[\cdot]$.
Of course, when $t=1$ the standard Boltzmann-Gibbs measure and related averages are recovered.

Analogously, we can also introduce a generalized interpolating quenched statistical pressures as 

\begin{definition} The interpolating  statistical pressure for the multitasking Hebbian neural network is introduced as
\begin{eqnarray}
\mathcal{A}^{(sup,unsup)}_{N,K\beta, d,M,r}(J, t) &\coloneqq& \frac{1}{N} \mathbb{E} \left[  \ln\mathcal{Z}^{(sup,unsup)}_{N,K, \beta, d,M,r}(\bm \eta; J, t)  \right],
\label{hop_GuerraAction}
\end{eqnarray}
and, in the thermodynamic limit,
\begin{equation}
\mathcal{A}^{(sup,unsup)}_{K,\beta, d,M,r}( J, t) \coloneqq \lim_{N \to \infty} \mathcal{A}^{(sup,unsup)}_{N,K, \beta, d,M,r}(J,t).
\label{hop_GuerraAction_TDL}
\end{equation}
Obviously, by setting $t=1$ in the interpolating pressures we recover the original ones, namely $\mathcal A^{(sup,unsup)}_{K, \beta, d,M,r} (J) = \mathcal{A}^{(sup,unsup)}_{K, \beta, d,M,r}(  J, t=1 )$, which we finally evaluate at $J=0$. 
%
%
%
\end{definition}


%
%
We are now ready to state the next 
\begin{theorem}
\label{prop:highnoise}
In the thermodynamic limit ($N\to\infty$) and in the low-storage regime ($K/N \to 0$), the quenched statistical pressure of the multitasking Hebbian network -- trained under supervised or unsupervised learning -- reads as 
\begin{equation}
\begin{array}{lll}
\mathcal{A}^{(sup,unsup)}_{K,\beta, d,M,r}(J)
 &=&\mathbb{E}\left\{\ln{}\Bigg[2\cosh\left(J\SOMMA{\mu=1}{K}\xi^\mu +\dfrac{\beta}{1-d}\SOMMA{\mu=1}{K}\n_\mu\hat{\eta}^\mu\right)\Bigg]\right\}-\dfrac{\beta}{1-d}(1+\rho)\SOMMA{\mu=1}{K}\n_\mu^2.
\end{array}
\end{equation}
\normalsize
where $\mathbb{E}=\mathbb{E}_\xi\mathbb{E}_{(\eta|\xi)}$, $\hat{\eta}^\mu=\dfrac{1}{Mr}\SOMMA{a=1}{M}\eta_i^{\mu,a}$, and the values $\bar{n}_\mu$  must fulfill the following self-consistent equations
\begin{equation}
    \begin{array}{lll}
        \n_\mu=\dfrac{1}{(1+\rho)}\mathbb{E}\left\{\tanh\left[\dfrac{\beta}{(1-d)}\SOMMA{\nu=1}{K}\n_\nu\hat{\eta}^\nu\right]\hat{\eta}^\mu\right\},\ \ \ \forall \mu \in (1,...,K),
        \label{eq:n_selfRS_sup} 
    \end{array}
\end{equation}
as these values of the order parameters are extremal for the statistical pressure $\mathcal{A}^{(sup,unsup)}_{K,\beta, d,M,r}(J=0)$.
\normalsize
\end{theorem}

\begin{corollary} \label{C1}
By considering the auxiliary field $J$ coupled to $m_{\mu}$ and recalling that $\lim_{N \to \infty}m_{\mu}=\bar{m}_{\mu}$, we can write down a self-consistent equation also for the Mattis magnetization as $\bar{m}_\mu=  \partial_J \mathcal{A}^{(sup,unsup)}_{K,\beta, d,M,r}(J)\vert_{J=0}$, thus we have
\begin{equation}
    \begin{array}{lll}
    \m_\mu = \mathbb{E} \left\{\tanh\left[\dfrac{\beta}{(1-d)}\SOMMA{\nu=1}{K}\n_\nu\hat{\eta}^\nu \right] \xi^\mu\right\},\ \ \  \forall \mu \in (1,...,K).
    \end{array}
    \label{eq:m_selfRS_sup}
\end{equation}
\normalsize
\end{corollary}
For the proof of Proposition \ref{prop:highnoise} and of Corollary \ref{C1} we refer to Appendix \ref{proof:highnoise}.


We highlighted that the expressions of the quenched statistical pressure for a network trained with or without the supervision of a teacher do actually coincide: intuitively, this happens because we are considering only a few archetypes (i.e. we work at low load), consequently, the minima of the cost function are well separated and there is only a negligible role of the teacher in shaping the landscape to avoid overlaps in their basins of attractions. Clearly, this is expected to be no longer true in the high load setting and, indeed, it is proven not to hold for non-diluted patterns, where supervised and unsupervised protocols give rise to different outcomes \cite{Agliari-Emergence,Aquaro-EPL}. 
By a mathematical perspective the fact that, whatever the learning procedure, the expression of the quenched statistical pressure is always the same, is a consequence of standard concentration of measure arguments \cite{Barbier,Talagrand}  as, in the $N\to\infty$ limit,  
beyond eq. (\ref{PdiEmme}), it is also  $\mathcal P(n_{\mu,a}) \to \delta(n_{\mu,a} - \bar{n}_{\mu})$. 
\newline
\newline
The self-consistent equations \eqref{eq:m_selfRS_sup} have been solved numerically for several values of parameters and results for K=2 and for $K=3$ are shown in  Fig.~\ref{fig:loss1} (where also the values of the cost function is reported)  and Fig.~\ref{fig:my_label} respectively. We also checked the validity of these results by comparing them with the outcomes of Monte Carlo simulations, finding an excellent asymptotic agreement; further, in the large $M$ limit, the magnetizations eventually converge to the values predicted by the theory developed in the storing framework, see eq. \eqref{eq:parallel_ans}.
Therefore, in both the scenarios, the hierarchical or parallel organization of the magnetization's amplitudes are recovered: beyond the numerical evidence just mentioned, in Appendix \ref{app:calc}  an analytical proof is provided.

\begin{figure}
    \centering
    \includegraphics[width=15.5cm]{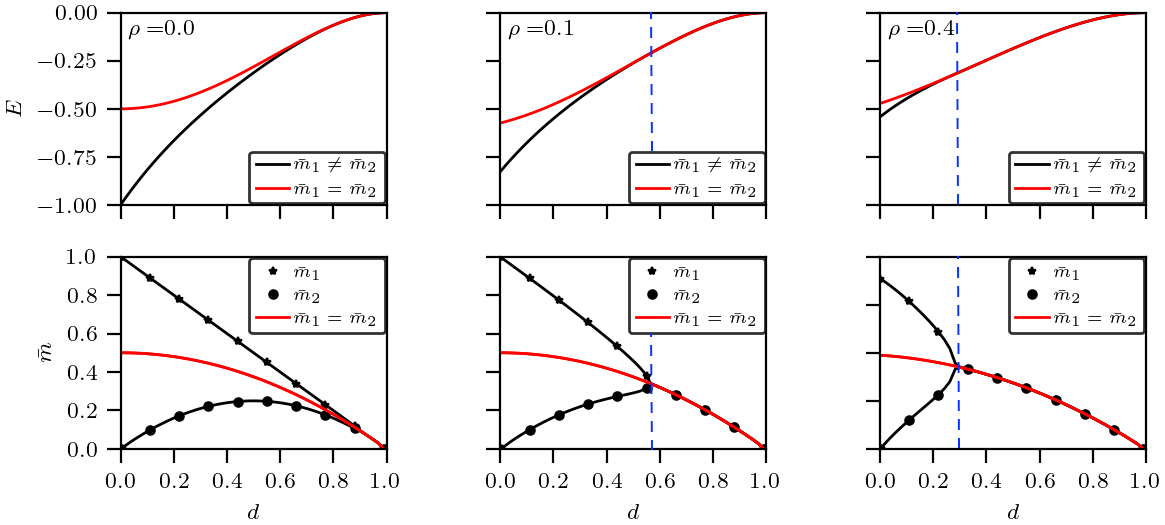}
    \caption{Snapshots of cost function (upper plots) -where we use the label $E$ for {\em energy}- and magnetizations (lower plots) for data-sets generated by $K=2$ archetypes and at different entropies,  in the noiseless limit $\beta \to \infty$. Starting by $\rho=0.0$ we see that the hierarchical regime (black lines) dominates at relatively mild dilution values (i.e., the energy pertaining to this configuration is lower w.r.t. the parallel regime), while for $d \to 1$ the hierarchical ordering naturally collapse to  the parallel regime (red lines), where all the magnetizations acquire the same values. Further note how, by increasing the entropy in the data-set (e.g. for $\rho=0.1$ and $\rho=0.4$), the domain of validity of the parallel regime enlarges (much as increasing $\beta$ in the network, see Fig. \ref{fig:UnoTaccitua}). The vertical blue lines mark the transitions  between these two regimes as captured by Statistical Mechanics: it corresponds to switching from the white to the green regions of the phase diagrams of Fig. \ref{fig:stabilità}. \label{fig:loss1}}
\end{figure}

\begin{figure}
    \centering
    \includegraphics[width=15cm]{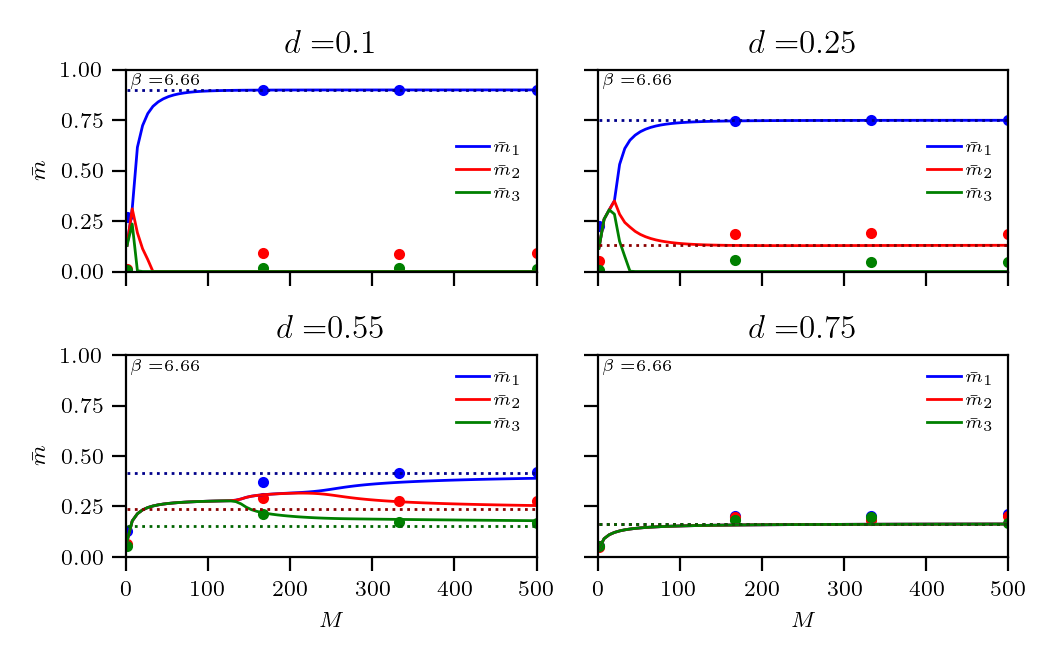}
    \caption{{\em Behaviour of the Mattis magnetizations as more and more examples are supplied to the network.}  Monte Carlo numerical checks (colored dots, $N=6000$) for a diluted network with $r=0.1$ and $K=3$ are
 in plain agreement with the theory: solutions of the self-consistent equation for the Mattis magnetizations reported in the Corollary \ref{C1} are shown as solid lines. As dilution increases, the network behavior departs from a Hopfield-like retrieval ($d=0.1$) where just the blue magnetization is raised (serial pattern recognition) to the hierarchical regime  ($d=0.25$ and $d=0.55$) where multiple patterns are simultaneously retrieved with different amplitudes, while for higher values of dilution the network naturally evolves toward the parallel regime ($d=0.75$) where all the magnetizations are raised and with the same strength.
 Note also the asymptotic agreement with the dotted lines, whose values are those predicted by the multitasking Hebbian storage \cite{Barra-PRL2012}. \label{fig:my_label}
}
\end{figure}


\subsubsection{Low-entropy data-sets: the {\em} Big-Data limit}
\label{sec:datalimit_sup}

As discussed in Sec.~\ref{Gramsci}, the parameter $\rho$ quantifies the amount of information needed to describe the original message
$\boldsymbol\xi^\mu$ given the set of related examples $\{\boldsymbol\eta^{\mu,a}\}_{a=1,...,M}$. In this section we focus on the case $\rho \ll 1$ that corresponds to a highly-informative data-set; we recall that in the limit $\rho \to 0$ we get a data-set where either the items ($r \to 1$) or their empirical average ($M \to \infty$, $r$ finite) coincide with the archetypes, in such a way that the theory collapses to the standard Hopfield reference. 
\newline
As explained in Appendix \ref{proof:prop_m_n}, we start from the self-consistent equations \eqref{eq:n_selfRS_sup}-\eqref{eq:m_selfRS_sup} and we exploit the Central Limit Theorem to write $\hat{\eta}^\mu\sim \xi^\mu\left(1+\lambda_\mu\sqrt{\rho}\right)$, where $\lambda_\mu\sim\mathcal{N}(0,1)$. In this way we reach the simpler expressions given by the next
\begin{proposition} \label{prop:m_n}
In the low-entropy data-set scenario, preserving the low storage and thermodynamic limit assumptions, the two sets of order parameters of the theory, $\bar{m}_{\mu}$ and $\bar{n}_{\mu}$ become related by the following equations
\begin{eqnarray}
   \n_\mu&=&\dfrac{\m_\mu}{(1+\rho)}+\b\dfrac{\rho\,\n_\mu}{(1+\rho)}\mathbb{E}_{\xi,Z}\left\{\left[1-\tanh^2\left(g(\beta,  Z, \bar{\bm n})\right)\right]\left(\xi^\mu\right)^2\right\},\label{eq:n_Mgrande}
    \\
    \notag
    \\
     \m_\mu &= &\mathbb{E}_{\xi,Z} \left\{\tanh\left[g(\beta, \bm \xi,  Z, \bar{\bm n})\right] \xi^\mu\right\},\label{eq:m_Mgrande}
\end{eqnarray}
where
\begin{align}
    &g(\beta,\bm \xi,  Z, \bar{\bm n})=\b\SOMMA{\nu=1}{K}\n_\nu\xi^\nu +\b\,Z\sqrt{\rho\,\SOMMA{\nu=1}{K}\n_\nu^2\left(\xi^\nu\right)^2}  \;
    \label{eq:g_of_super_n}
\end{align}
and $Z \sim \mathcal{N}(0,1)$ is a standard Gaussian variable. Furthermore, to lighten the notation and assuming $d \neq 1$ with no loss of generality, we posed
\begin{equation}
\b = \dfrac{\beta}{1-d}.
\label{eq:tildebeta}
\end{equation}
\end{proposition}
The regime $\rho \ll 1$, beyond being an interesting one (e.g., it can be seen as a {\em big data} $M \to \infty$ limit of the theory), offers a crucial advantage because of the above emerging proportionality relation between $\bar n$ and $\bar m$ (see eq. \ref{eq:n_Mgrande}). In fact, the model is supplied only with examples -- upon which the $n_{\mu}$'s are defined -- while it is not aware of archetypes  -- upon which the $m_{\mu}$'s are defined -- yet we can use this relation to recast  the self-consistent equation for $\bar n$  into a self-consistent equation for $\bar m$ such that its numerical solution in the space of the control parameters allows us to get the phase diagram of such a neural network more straightforwardly. 
\newline
Further, we can find out explicitly the thresholds for learning, namely the minimal amount of examples (given the level of noise $r$, the amount of archetype to handle $K$, etc.) that guarantee that the network can safely infer the archetype from the supplied data-set. To obtain  these thresholds we have to deepen the ground state structure of the network,  that is, we now  handle the Eqs.~\eqref{eq:n_Mgrande}-\eqref{eq:m_Mgrande} to compute their zero fast-noise limit ($\beta \to \infty$). 
%
%
%
%
As detailed in the Appendix \ref{proof:prop_m_n} (see Corollary \ref{Corollario3}), by taking the limit $\beta \to \infty$ in eqs.~\eqref{eq:n_Mgrande}-\eqref{eq:m_Mgrande} we get
\begin{equation}
    \begin{array}{lll}
        \m_\mu &=&   \mathbb{E}_{\xi} \left\{\mathrm{erf}\left[\left(\sum\limits_{\nu=1}^{K}\m_\nu\xi^\nu\right)\left(2\rho\sum\limits_{\nu=1}^{K}\m_\nu^2\left(\xi^\nu\right)^2  \right)^{-1/2} \right] \xi^\mu\right\}\,.
    \end{array}
    \label{eq:noiseless}
\end{equation}

Once reached a relatively simple expression for $\bar m_{\mu}$, we can further manipulate it and try to get information about the existence of a lower-bound value for $M$, denoted with $M_{\otimes}$, which ensures that the network has been supplied with sufficient information to learn and retrieve the archetypes.
\newline
Setting $\beta \to \infty$, we expect that the magnetizations fulfill the hierarchical organization, namely $(\m_1,\m_2,\hdots) = (1-d)(1,d,\hdots)$ and \eqref{eq:noiseless} becomes
\begin{equation}
    \begin{array}{lll}
         \m_\mu\sim\dfrac{1-d}{2}\mathbb{E}_{_{\xi^{\nu\neq\mu}}} \left\{\mathrm{erf}\left[\frac{d^{\mu}+\sum\limits_{\nu\neq\mu}^{K}d^{\nu}\xi^\nu}{\sqrt{2\rho}\sqrt{d^{2{\mu}}+\sum\limits_{\nu\neq\mu}^{K}d^{2\nu}\left(\xi^\nu\right)^2}} \right]+\mathrm{erf}\left[\frac{d^{\mu}-\sum\limits_{\nu\neq\mu}^{K}d^{\nu}\xi^\nu}{\sqrt{2\rho}\sqrt{d^{2\mu}+\sum\limits_{\nu\neq\mu}^{K}d^{2\nu}\left(\xi^\nu\right)^2}} \right] \right\}\,,
    \end{array}
\end{equation}
where we highlighted that the expectation is over all the archetypes but the $\mu$-th one under inspection. 

Next, we introduce a confidence interval, ruled by $\Theta$, and we require that 
\begin{equation}
    \m_\mu >(1-d)d^{\mu-1} \mathrm{erf}\left[\Theta\right].
    \label{eq:inequality}
\end{equation}
In order to quantify the critical number of examples $M_\otimes^{\mu}$ needed for a successful learning of the archetype $\mu$ we can exploit the relation
\begin{equation}
    \begin{array}{lll}
         \mathbb{E}_{_{\xi^{\nu\neq\mu}}} \Big\{\mathrm{erf}\Big[f(\xi)\Big] \Big\}\geq\underset{\xi^{\nu\neq\mu}}{\mathrm{min}} \Big\{\mathrm{erf}\Big[f(\xi)\Big] \Big\}\,,
    \end{array}
\end{equation}
where in our case
\begin{equation}
\footnotesize
    \begin{array}{lll}
        \underset{\xi^{\nu\neq\mu}}{\mathrm{min}} \Big\{\mathrm{erf}\Big[f(\xi)\Big] \Big\}&=&\underset{\xi^{\nu\neq\mu}}{\mathrm{min}} \left\{\mathrm{erf}\left[\dfrac{d^{\mu}+\sum\limits_{\nu\neq\mu}^{K}d^{\nu}\xi^\nu}{\sqrt{2\rho}\sqrt{d^{2{\mu}}+\sum\limits_{\nu\neq\mu}^{K}d^{2\nu}\left(\xi^\nu\right)^2}} \right]+\mathrm{erf}\left[\dfrac{d^{\mu}-\sum\limits_{\nu\neq\mu}^{K}d^{\nu}\xi^\nu}{\sqrt{2\rho}\sqrt{d^{2\mu}+\sum\limits_{\nu\neq\mu}^{K}d^{2\nu}\left(\xi^\nu\right)^2}} \right] \right\}
        \\\\
        &=&2\;\mathrm{erf}\left[\left(d^{\mu}-\sum\limits_{\nu\neq\mu}^{K}d^{\nu}\right)\left(2\rho\sum\limits_{\nu=1}^{K}d^{2\nu}\right)^{-1/2}\right].
    \end{array}
\end{equation}
Thus, using the previous relation in  \eqref{eq:inequality}, the following inequality must hold
\begin{equation}
    \begin{array}{lll}
         \mathrm{erf}\left[\left(d^{\mu}-\sum\limits_{\nu\neq\mu}^{K}d^{\nu}\right)\left(2\rho\sum\limits_{\nu=1}^{K}d^{2\nu}\right)^{-1/2}\right]=\mathrm{erf}\left[\sqrt{\dfrac{1+d}{2\rho(1-d)}}\dfrac{2d^{\mu-1}-1-2d^\mu+d^K}{\sqrt{1-d^{2K}}} \right] > d^{\mu-1}\mathrm{erf}\left[\Theta\right] 
    \end{array}
\end{equation}
and we can write the next
\begin{proposition}
In the noiseless limit $\beta \to \infty$, the critical threshold for learning $M_\otimes$ (in the number of required examples) depends on the data-set noise $r$, the dilution $d$, the amount of archetypes to handle $K$ (and of course on the amplitude of the chosen confidence interval $\Theta$) and reads as  
\begin{equation}
    \begin{array}{lll}
     M_\otimes^\mu(\Theta, r, d, K)  > 2\left(\mathrm{erf}^{-1}\left[d^{\mu-1}\mathrm{erf}\left[\Theta\right]\right]\right)^2\left(\dfrac{1-r^2}{r^2}\right)\dfrac{(1-d)(1-d^{2K})}{(1+d)(2d^{\mu-1}-1-2d^{\mu}+d^K)^2}\,
     \label{eq:critical_M}
    \end{array}
\end{equation}
and in the plots (see Figure \ref{fig:M_otimes}) we use $\Theta = 1/\sqrt{2}$ as this choice corresponds to the fairly standard condition $\mathbb{E}_{\xi}\mathbb{E}_{(\eta|\xi)}[\xi_i^1h_i^{(1)}(\bm\xi^1)]>\sqrt{\mathrm{Var}[\xi_i^1h_i^{(1)}(\bm\xi^1)]}$ when $\mu=1$.
\end{proposition}
To quantify these thresholds for learning, in Fig. \ref{fig:M_otimes} we report the required number of examples to learn the first archetype (out of $K=2,3,4,50$ as shown in the various panels) as a function of the dilution of the network.

\begin{figure}
    \centering
    \includegraphics[width=15cm]{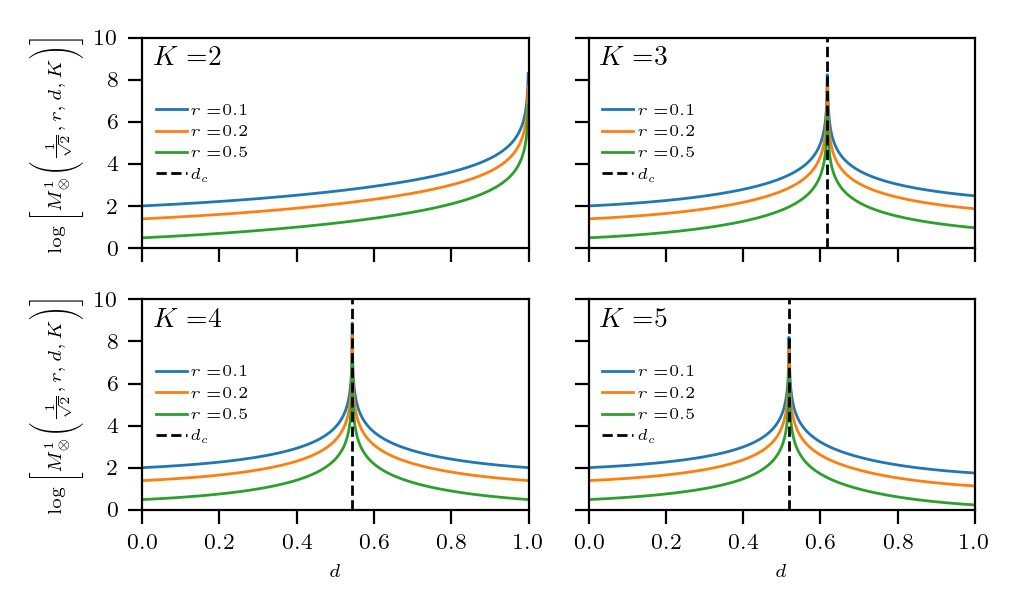}
    \caption{We plot the logarithm of the critical number of example (required to raise the first magnetization)  $M_\otimes^1$ at different loads $K=2,3,4,5$ and as a function of the dilution of the networks, for different noise values of the data-set (as shown in the legend). Note the divergent behavior of $M_\otimes^1$ when approaching the critical dilution level $d_c(K)=d_1$, as predicted by the parallel Hebbian storage limit \cite{Barra-PRL2012,ElenaMult1}: this is the crossover between the two multi-tasking regimes, hierarchical vs parallel, hence, solely at the value of dilution $d_1$, there is no sharp behavior to infer and, correctly, the network can not accomplish learning. This shines by looking at \eqref{eq:critical_M} where the critical amount of examples to correctly infer the archetype is reported: its denominator reduces to $1-2d+d^K$ and, for $\mu=1$, it becomes zero when $d \to d_1$.}
    \label{fig:M_otimes}
\end{figure}
 


\subsubsection{Ergodicity breaking: the critical phase transition}\label{cazzopeppe}

The main interest in the statistical mechanical 
approach to neural networks lies in inspecting their emerging capabilities, that typically appear once ergodicity gets broken: as a consequence, finding the boundaries of validity of ergodicity is a classical starting point to deepen these aspects. 
\newline
To this task, hereafter we provide a systematic fluctuation analysis of the order parameters: the underlying idea is to check when, starting from the high noise limit ($\beta \to 0$, where everything is uncorrelated and simple Probability arguments apply straightforwardly), these  fluctuations diverge, as that  defines the onset of ergodicity breaking as stated in the next
\begin{theorem}\label{Fluttuo}
The ergodic region, in the space of the control parameters $(\beta, d, \rho)$ is confined to the half-plane defined by the critical line 
\begin{equation}
\beta_c =\frac{1}{1-d}, \end{equation}
whatever the entropy of the data-set $\rho$.
\end{theorem}
\begin{proof}
The idea of the proof is the same we used so far, namely Guerra interpolation but on the rescaled fluctuations rather that directly on the statistical pressure.
\newline
The rescaled fluctuations  $\tilde{n}^2_{\nu}$ of the magnetizations are defined as  
\begin{equation}
\tilde n_{\nu}= \sqrt{N}(n_\nu-\bar{n}_{\nu}).
\end{equation}
We remind that the interpolating framework we are using, for $t\in (0,1)$, is defined via
\begin{equation}
    Z(t) = \SOMMA{\{\bm\sigma\}}{} \exp\left[\dfrac{\beta}{2}tN(1+\rho) \SOMMA{\mu=1}{K}n_{\mu}^2+ N(1-t)\beta(1+\rho) N_\mu n_{\mu}\right],
\end{equation}
and it is a trivial exercise to show that, for any smooth function $F(\sigma)$ the following relation holds:
\begin{equation}
   \dfrac{d \l F \r}{dt} = \dfrac{\beta}{2}(1+\rho)\left(\l F \SOMMA{\nu}{}\tilde n_{\nu}^2\r  -\l F\r\l \SOMMA{\nu}{}\tilde n_{\nu}^2\r \right),
\end{equation}
such that by choosing $F=\tilde n^2_\mu$ we can write
\begin{equation}
\begin{array}{lll}
     \dfrac{d \l \tilde n_\mu^2 \r}{dt} &= \dfrac{\beta}{2}(1+\rho)\left(\l \tilde n_\mu^2 \SOMMA{\nu}{}\tilde n_{\nu}^2\r  -\l \tilde n_\mu^2\r\l \SOMMA{\nu}{} \tilde n_{\nu}^2\r \right)
     \\\\
     &=\dfrac{\beta}{2}(1+\rho)\left(\l
 \tilde n_\mu^4\r+\l \n_\mu^2 \SOMMA{\nu\neq\mu}{}\tilde n_{\nu}^2\r  -\l \tilde n_\mu^2\r^2-\l \tilde n_\mu^2\r\l \SOMMA{\nu\neq\mu}{}\tilde n_{\nu}^2\r \right)
      \\\\
     &=\beta(1+\rho)\l \tilde n_\mu^2\r^2
\end{array}
\end{equation}
thus we have
\begin{equation}
\begin{array}{lll}
      \l \tilde n_\mu^2 \r_t &= \dfrac{\l
 \tilde n_\mu^2 \r_{t=0}}{1-t\beta(1+\rho)\l \tilde n_\mu^2 \r_{t=0}}
\end{array}
\end{equation}
where the Cauchy condition $\l \tilde n_\mu^2 \r_{t=0}$ reads
\begin{equation}\small
\begin{array}{lll}
      \l \tilde n_\mu^2 \r_{t=0}&=& N \mathbb E_{\xi} \mathbb E_{(\eta|\xi)}\dfrac{\SOMMA{\{\bm\sigma\}}{}\left(\frac{1}{N^2(1+\rho)^2}\sum_{i,j}\tilde\eta_i^\mu\tilde\eta_j^\mu\sigma_i\sigma_j+\tilde n_\mu^2-2\frac{1}{N(1+\rho)}\sum_i \hat\eta_i^\mu\sigma_i \tilde n_\mu\right)\exp\left[ \beta\sum_\nu \tilde n_\nu\sum_i\tilde\eta_i^\nu\sigma_i\right]}{\SOMMA{\{\bm\sigma\}}{}\exp\left[ \beta\sum_\nu \tilde n_\nu\sum_i\tilde\eta_i^\nu\sigma_i\right]}
    \\\\
    &=&  \dfrac{1-d}{(1+\rho)}-N_\mu^2.
\end{array}
\end{equation}
Evaluating $\l \tilde n_\mu^2 \r_{t}$ for $t=1$, that is when the interpolation scheme collapses to the Statistical Mechanics, we finally get
\begin{equation}
\begin{array}{lll}
      \l \tilde n_\mu^2 \r_{t=1 }&= \dfrac{1-d-(1+\rho)N_\mu^2}{[1-\beta\left(1-d-(1+\rho)N_\mu^2\right)]}
\end{array}
\end{equation}
namely the rescaled fluctuations are described by a meromorphic function whose pole is
\begin{equation}
    \beta=\dfrac{1}{\left(1-d-(1+\rho)N_\mu^2\right)}\xrightarrow[]{N_\mu=0}\beta_{_C}=\dfrac{1}{1-d},
\end{equation}
that is the critical line reported in the statement of the theorem.
\end{proof}

\subsection{Stability analysis via standard Hessian: the phase diagram}\label{cazzarola}

The set of solutions for the self-consistent equations for the order parameters \eqref{eq:n_Mgrande} describes as candidate solutions a plethora of states whose stability must be investigated to understand which solution is preferred as the control parameters are made to vary: this procedure results in picturing the phase diagrams of the network, namely plots in the space of the control parameters where different regions pertain to different macroscopic computational capabilities. 
\newline
Remembering that $A_{K,\beta,d,M,r}(\bar{\bm n}) =
-\beta f_{K,\beta,d,M,r}(\bar{\bm n})$ (where $f_{K,\beta,d,M,r}(\bar{\bm n})$ is the free energy of the model), in order to evaluate the stability of these solutions, we need to check the sign of the second derivatives of the free energy. More precisely, we  need to build up the Hessian,  a matrix $\bm A$ whose elements are
\begin{equation}
    \dfrac{\partial^2 f (\bar{\bm n})}{\partial n^{\mu}\partial n^\nu}=A^{\mu\nu}\,.
\end{equation}
Then, we evaluate and diagonalize $\bm A$ at a point $\tilde{\bm n}$, representing a particular solution of the self-consistency equation \eqref{eq:n_Mgrande}: the numerical results are reported in the phase diagrams provided in Fig.\ref{fig:stabilità}.

We find straightforwardly
\begin{equation}\label{Fhesso}
    A^{\mu\nu} =  (1+\rho)\left[[1-\beta(1-d)]+\rho\beta\mathbb{E}\left\{\mathcal{T}_{K\beta,\rho}^2(\bm{\n},z)(\xi^\mu)^2\right\}\right]\delta^{\mu\nu} + Q^{\mu\nu}
\end{equation}
where we set $\mathcal{T}_{K\beta,\rho}(\bm{\n},z)=\tanh\left(\beta\sum_{\lambda=1}^{K}\n_\lambda\xi^\lambda+z\beta\sqrt{\rho\sum_{\lambda=1}^{K}(\n_\lambda\xi^\lambda)^2}\right)$ and
\begin{equation}
\footnotesize
    \begin{array}{lll}
        Q^{\mu\nu}&=&\beta\mathbb{E}\left\{\left[\mathcal{T}_{K\beta,\rho}^2(\bm{\n},z)\right]\xi^\mu\xi^\nu\right\}(1-\delta^{\mu\nu})+2\rho\beta^2\mathbb{E}\left\{\left[\mathcal{T}_{K\beta,\rho}(\bm{\n},z)\right]\left[1-\mathcal{T}_{K\beta,\rho}^2(\bm{\n},z)\right]\left[\n_\nu\xi^\nu+\n_\mu\xi^\mu\right]\xi^\mu\xi^\nu\right\}
        \\\\
    &&+2\rho^2\beta^3\n_\mu\n_\nu\mathbb{E}\left\{\left[1-3\mathcal{T}_{K\beta,\rho}^2(\bm{\n},z)\right]\left[1-\mathcal{T}_{K\beta,\rho}^2(\bm{\n},z)\right](\xi^\mu\xi^\nu)^2\right\}
    \end{array}
\end{equation}
In order to lighten the notation we will use $\mathcal{T}_{K\beta,\rho}(\bm{\n},z)=\mathcal{T}_K$. 
\begin{figure}
    \centering
    \includegraphics[width=13cm]{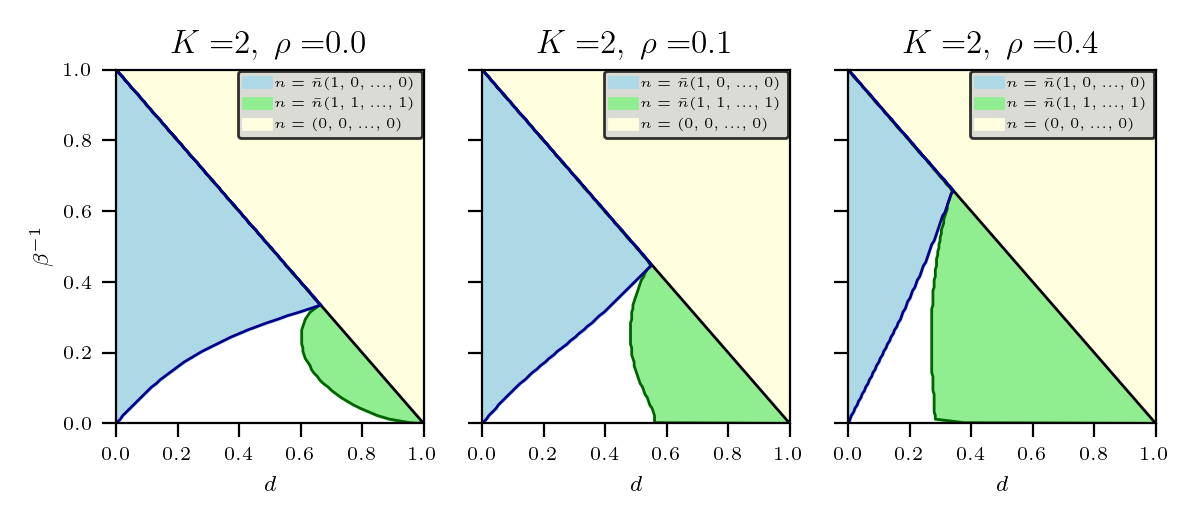}
    \includegraphics[width=13cm]{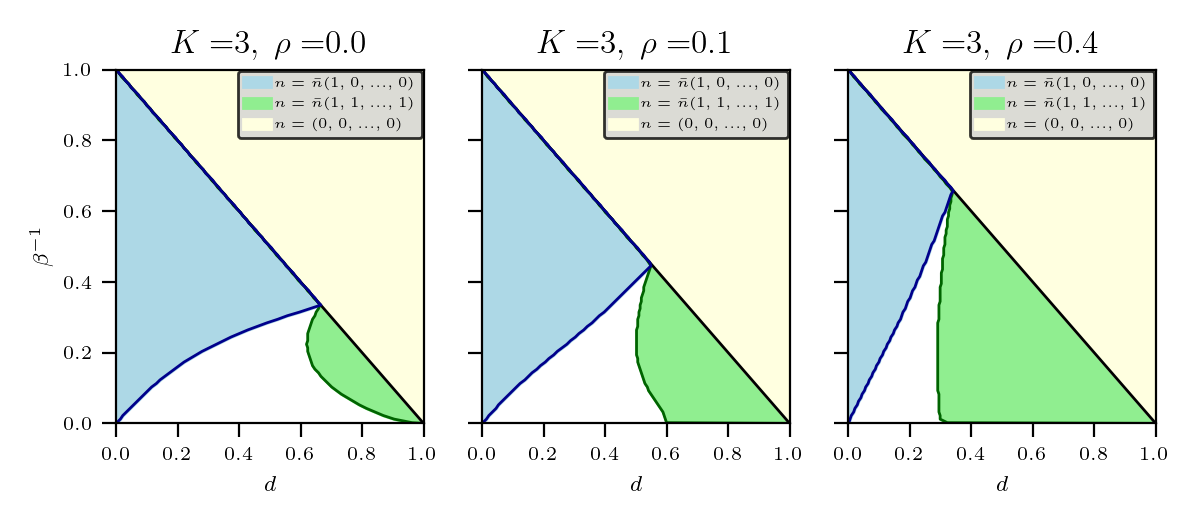}
    \caption{Phase diagram in  the dilution-noise $(d-\beta^{-1})$ plane for different values of $K$ and $\rho$. We highlight that different regions --marked with different colors-- represent different operational behavior of the network: in yellow the ergodic solution, in light-blue the pure state solution (that is, solely one magnetization different from zero), in white the hierarchical regime (that is, several magnetizations differ from zero and they all assume different values)  and in  light-green the parallel regime (several magnetization differ from zero but their amplitude is the same for all). }
    \label{fig:stabilità}
\end{figure}
We can now inspect the domain of stability of each possible solution of the self-consistency equation by plugging the structure of the candidate solution in \eqref{Fhesso}.

\subsubsection{Ergodic state: $ \bar{\bm n}=\n_{d,\rho,\beta}( 0,  \hdots, 0)$}
In this case the structure of the solution has the form $\bm \n =\bm \m = \bm 0$ thus the Hessian matrix is diagonal and it reads as 
\begin{equation}
    A^{\mu\nu}= \delta^{\mu\nu}(1+\rho)[1-\beta(1-d)].
\end{equation}
As all its eigenvalues are equal to $(1+\rho)[1-\beta(1-d)]$, if we require the matrix to be positively definite we must have
\begin{equation}
    d>\dfrac{\beta-1}{\beta}.
\end{equation}
Therefore, as $d>1-\beta^{-1}$, the ergodic solution is stable: this scenario is reported in the phase diagrams provided in Fig. \ref{fig:stabilità} as the yellow region.
\newline
We stress that this result on the ergodic region is in plain agreement with the inspection on ergodicity breaking provided in Theorem \ref{Fluttuo}.

\subsubsection{Pure state: $\bar{\bm n} = \n_{d,\rho,\beta}(1, 0 , \hdots, 0)$}
In this case  the structure of the solution has the form $\m_{\mu}=\n_{\mu}=0$ for $\mu>1$, thus the only self-consistency equation different from zero is
\footnotesize
\begin{eqnarray}
   &\n=\dfrac{\mathbb{E}_{\xi,Z}\left\{\left[\mathcal{T}_K\right]\left(\xi^\mu\right)\right\}}{(1+\rho)}+\beta\dfrac{\rho\,\n}{(1+\rho)}\mathbb{E}_{\xi,Z}\left\{\left[1-\mathcal{T}_K^2\right]\left(\xi^\mu\right)^2\right\}\,,
\end{eqnarray}
\normalsize
where $\mathcal{T}=\tanh\left[\beta\n\xi^\mu(1+z\sqrt{\rho})\right]$. It is easy to check that 
$\bm A$ becomes diagonal, with
\begin{equation}
\small
\begin{array}{lll}
A^{\mu\mu}&=&(1+\rho)\Big[1-\beta(1-d)+\beta(1-d)\,\mathbb{E}\left\{\mathcal{T}^2\right\}\Big]
    \\\\
&& +4\beta^2 \rho \n (1-d)\mathbb{E}\left\{\mathcal{T}\Big[1-\mathcal{T}^2\Big]\right\}+2\beta^3 \rho^2 \n^2 (1-d)\mathbb{E}\left\{\Big[1-3\mathcal{T}^2\Big]\Big[1-\mathcal{T}^2\Big]\right\}\,,
\\\\
A^{\nu\nu\neq \mu}&=&(1+\rho)\Big[1-\beta(1-d)+\beta(1-d)^2\,\mathbb{E}\left\{\mathcal{T}^2\right\}\Big] \,.
\end{array}
\end{equation}
Notice that these eigenvalues do not depend on $K$ since $\mathcal{T}$ does not depend on $K$. Requiring
the positivity for all the eigenvalues, we get the region in the plane $(d , \beta^{-1})$, where the pure state is stable: this correspond the blue region in the phase diagrams reported in Fig. \ref{fig:stabilità}.
\newline
We stress that these pure state solutions, namely the standard Hopfield-type ones, in the ground state ($\beta^{-1} \to 0$) are never stable whenever $d \neq 0$ as the multi-tasking setting prevails. Solely at positive values of $\beta$, this single-pattern retrieval state is possible as the role of the noise is to destabilize the weakest magnetization of the hierarchical displacement, {\em vide infra}).

\subsubsection{Parallel state: $\bar{\bm n} =\n_{d,\rho,\beta} (1, \hdots, 1)$}
In this case  the structure of the solution has the form of a symmetric mixture state corresponding to the unique self consistency equation for all $\mu=1,\hdots K$, namely
\begin{eqnarray}
\n&=&\dfrac{\mathbb{E}_{\xi,Z} \left\{\tanh\left[g(\beta, \bm \xi,  Z, \n)\right] \xi^\mu\right\}}{(1+\rho)}+\beta\dfrac{\rho\,\n}{(1+\rho)}\mathbb{E}_{\xi,Z}\left\{\left[1-\tanh^2\left(g(\beta,\bm\xi  Z, \n)\right)\right]\left(\xi^\mu\right)^2\right\},
\end{eqnarray}
where
\begin{align}
    &g(\beta,\bm \xi,  Z, \n)=\beta\n\left[\SOMMA{\lambda=1}{K}\xi^\lambda +\beta\,Z\sqrt{\rho\,\SOMMA{\lambda=1}{K}\left(\xi^\lambda\right)^2}\right]  \;.
\end{align}
In this case, the diagonal terms of $\bm A$ are
\begin{equation}
\small
\begin{array}{lll}
a=A^{\mu\mu}&=&\Big[1-\beta(1-d)+\beta\,\mathbb{E}\left\{\left[\mathcal{T}^2\right](\xi^\mu)^2\right\}\Big](1+\rho) 
    \\\\
&& +4\beta^2\rho\n\mathbb{E}\left\{\mathcal{T}\Big[1-\mathcal{T}^2\Big]\xi^\mu\right\}+2\beta^3\rho^2\n^2\mathbb{E}\left\{\Big[1-3\mathcal{T}^2\Big]\Big[1-\mathcal{T}^2\Big](\xi^\mu)^2\right\}\,,
\end{array}
\label{eq:a_stab}
\end{equation}
instead the off-diagonal ones are
\begin{equation}
\small
\begin{array}{lll}
b=A^{\mu\nu\neq\mu}&=&\beta\mathbb{E}\left\{\left[\mathcal{T}^2\right]\xi^\mu\xi^\nu\right\}+2\rho^2\beta^3\n^2\mathbb{E}\left\{\left[1-3\mathcal{T}^2\right]\left[1-\mathcal{T}^2\right](\xi^\mu\xi^\nu)^2\right\}\,.
\end{array}
\label{eq:b_stab}
\end{equation}
In general the following relationship holds
\begin{equation}
    \bm A =  \begin{pmatrix} a & b &\cdots &b& b
    \\
    b & a &  \cdots    &b & b
    \\
    \vdots& \vdots&\ddots &\vdots   &\vdots
    \\
    b & b &  \cdots     &a& b
    \\
    b & b & \cdots & b&a
    \end{pmatrix}
\end{equation}
This matrix has always only two kinds of eigenvalues, namely $a-b$ and $a+(K-1)b$, thus, for the stability of the parallel  state,  after computing \eqref{eq:a_stab} and \eqref{eq:b_stab}, we have only to check for which point in the $(d-\beta^{-1})$ plane both $a-b$ and $a+(K-1)b$ are positive. The region, in the phase diagrams of Fig. \ref{fig:stabilità}, where the parallel regime is stable is depicted in green.

\subsubsection{Hierarhical state: $\bar{\bm n} =\n_{d,\rho,\beta} ((1-d),d(1-d),d^2(1-d),...)$}

In this case  the structure of the solution has the hierarchical form $\bar{\bm n} =\n_{d,\rho,\beta} ((1-d),d(1-d),d^2(1-d),...)$ and the region left untreated so far in the phase diagram, namely the white region in the plots of Fig. \ref{fig:stabilità}, is the room left to such hierarchical regime.

\subsection{From the Cost function to the Loss function}\label{SezioneTreTre}
\begin{figure}[!]
    \centering
    \includegraphics[width=7.5cm]{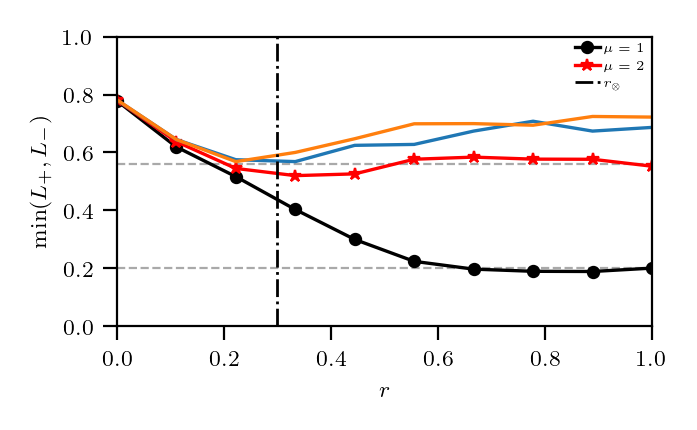}
   \includegraphics[width=7.5cm]{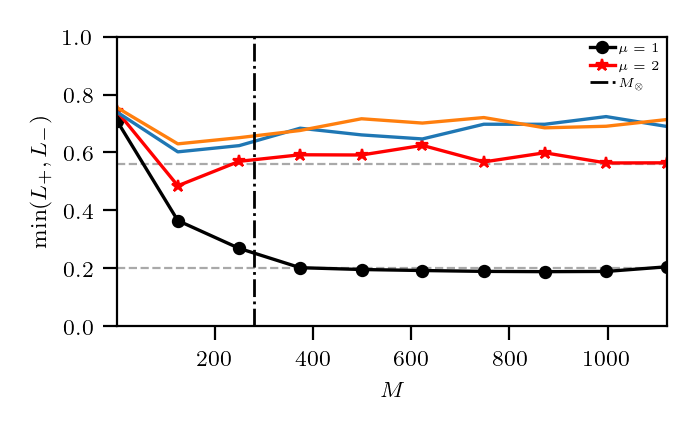} 
   \caption{Left: Parallel minimization of several (mean square-error) Loss functions $L_{\pm} = ||\xi^{\mu}\pm\sigma||^2$ (each pertaining to a different archetype) as the noise in the data-set $r$ is varied. Here: $M=25$, $N=10000$. 
   The horizontal gray dashed lines are the saturation level of the Loss functions, namely $1 - \frac{d}{2} - (1 - d)  d^{\mu-1}$. We get $r_\otimes$ (the vertical black line) by the inversion of \eqref{eq:critical_M}.  Right: Parallel minimization of several (mean square-error) Loss function $L_{\pm} = ||\xi^{\mu}\pm\sigma||^2$ (each pertaining to a different archetype) as the data-set size $M$ is varied: as M grows the simultaneous minimization of more than one Loss functions takes place, at difference with learning via standard Hebbian mechanisms where one Loss function -dedicate to a single archetype- is minimized per time. Orange and blue lines pertain to Loss functions of other patterns that, at these levels of dilution and noise, can not be minimized at once with the previous ones. \label{fig:loss2}}
\end{figure}
We finally comment on the persistence -in the present approach- of quantifiers related to the evaluation of the pattern recognition capabilities of neural networks, i.e. the Mattis magnetization, also as quantifiers of a good learning process. The standard Cost functions used in Statistical Mechanics of neural networks (e.g., the Hamiltonians) can be related one-to-one to standard Loss functions used in Machine Learning (i.e. the squared sum error functions), namely, once introduced the two Loss functions $L_{\mu}^+ := (1/2N)||\xi^{\mu} - \sigma||^2 = 1 - m_{\mu}$ and  $L_{\mu}^- = (1/2N)||\xi^{\mu} + \sigma||^2= 1 + m_{\mu}$\footnote{Note that in the last passage we naturally highlighted the presence of  the Mattis magnetization in these Loss functions.}, it is immediate to show that
$$
H(\boldsymbol{\sigma}|\boldsymbol{\xi}) = \frac{-1}{2N}\sum_{i,j}^{N,N} \sum_{\mu}^K \xi_i^{\mu}\xi_j^{\mu} \sigma_i \sigma_j \equiv - N \sum_{\mu}^K \left( 	1 - L_{\mu}^+ \cdot L_{\mu}^-\right),
$$ 
thus minimizing the former implies minimizing the latter such that, if we are extremizing w.r.t. the neurons we are performing machine retrieval (i.e. pattern recognition), while if we extremize w.r.t. the weights we perform machine leaning: indeed, at least in this setting, learning and retrieval are two faces of the same coin (clearly the task here, from a machine learning perspective, is rather simple as the network is just asked to correctly classify the examples and possibly generalize).
 
In Fig. \ref{fig:loss2} we inspect what happens to these Loss functions -pertaining to the various archetypes- as the Cost function gets minimized: we see that, at difference with the standard Hopfield model (where solely one Loss function per time diminishes its value), in this parallel learning setting several Loss functions (related to different archetypes) get simultaneously lowered, as expected by a parallel learning machine.

\section{Conclusions} \label{SezioneConclusiva}
Since the AGS milestones on Hebbian learning dated 1985 \cite{AGS1,AGS2}, namely the first comprehensive statistical mechanical theory of the Hopfield model for pattern recognition and associative memory, attractor neural networks have experienced an 
unprecedented growth and the bulk of techniques developed for spin glasses in these four decades  (e.g. replica trick, cavity method, message passage, interpolation) acts now as a prosperous cornucopia for explaining the emergent information processing capabilities that these networks show as their control parameters are made to vary.
\newline
In these regards, it is important to stress how nowadays it is mandatory to optimize AI protocols (as machine learning for complex structured data-sets is still prohibitively expensive in terms of energy consumption \cite{MIT_CO2}) and, en route toward a Sustainable AI (SAI), statistical mechanics may still pave a main theoretical strand: in particular, we highlight how the knowledge of the {\em phase diagrams} related to a given neural architecture (that is the ultimate output product of the statistical mechanical approach)  allows to set "a-priori" the machine in the optimal working regime for a given task thus unveiling a pivotal role of such a methodology even for a conscious usage of AI (e.g. it is useless to force a standard Hopfield model beyond its critical storage capacity)\footnote{Further, still searching for optimization of resources, such a theoretical approach can also be used to inspect key computational shortcuts (as, e.g. early stopping criteria faced in Appendix \ref{subsec:S2N_sup} or the role of  the size of the mini-batch used for training \cite{FedeNew} or of the flat minima emerging after training \cite{Riccardo}.}.
\newline
Focusing on Hebbian learning, however, while the original AGS theory remains a solid pillar and a paradigmatic reference in the field, several extensions are required to keep it up to date to deal with modern challenges: the first generalization we need is to move from a setting where the machine stores already defined patterns (as in the standard Hopfield model) toward a more realistic learning procedure where these patterns are unknown and have to be inferred from examples: the Hebbian storage rule of AGS theory quite naturally generalizes toward both supervised and an unsupervised learning prescriptions \cite{Agliari-Emergence,Aquaro-EPL}. This enlarges the space of the control parameters from $\alpha, \beta$ (or $K$, $N$, $\beta$) of the standard Hopfield model toward $\alpha, \beta, \rho$ (or $K$, $N$, $\beta$, $M$, $r$) as we now deal also with a data-set where we have $M$ examples of mean quality r for each pattern (archetype) or, equivalently, we speak of a data-set produced at given entropy $\rho$.
\newline
Once this is accomplished, the second strong limitation of the original AGS theory that must be relaxed is that patterns share the same length and, in particular, this equals the size of the network (namely in the standard Hopfield model there are N neurons to handle patterns, whose length is exactly N for all of them): a more general scenario is provided by dealing with patterns that contain different amounts of information, that is patterns are diluted. Retrieval capabilities of the Hebbian setting at work with diluted patterns have been extensively investigated in the last decade \cite{Barra-PRL2012,Barra-PRL2014,Barra-PRL2015,Coolen-Medium,Coolen-High,Remi1,Remi2,AuroRev,Huang} and it has been understood how, dealing with patterns containing sparse entries, the network automatically is able to handle several of them in parallel (a key property of neural networks that is not captured by standard AGS theory). However the study of the parallel learning of diluted patterns was not addressed in the Literature and in this paper we face this problem, confining this first study to the low storage regime, that is when the number of patterns scales at most logarithmically with the size of the network. Note that this further enlarges the space of the control parameters by introducing the dilution $d$: we have several control parameters because the network information processing capabilities are enriched w.r.t. the bare Hopfield reference\footnote{However we clarify how it could be inappropriate to speak about structural differences among the standard Hopfield model and the present multitasking counterpart: ultimately these huge behavioral differences are just consequences of the different nature of the the data-sets provided to the same Hebbian network during training.}.
\newline
We have shown here that if we supply to the network a data-set equipped with dilution, namely a sparse data-set whose patterns contain -on average- a fraction $d$ of blank entries (whose value is 0) and, thus, a fraction $(1-d)$ of informative entries (whole values can be $\pm 1$), then the network spontaneously undergoes parallel learning and behaves as a multitasking associative memory able to learn, store and retrieve multiple patterns in parallel. Further, focusing on neurons, the Hamiltonian of the model plays as the Cost function for neural dynamics, however, moving the attention to synapses, we have shown how the latter is one-to-one related to the standard (mean square error) Loss function in Machine Learning and this resulted crucial to prove that, by experiencing a diluted data-set, the network lowers in parallel several Loss functions (one for each pattern that is learning from the experienced examples). 
\newline
For mild values of dilution, the most favourite displacement of the Mattis magnetizations is a {\em hierarchical ordering}, namely the intensities of these signals scale as power laws w.r.t. their information content $m_K \sim d^K \cdot (1 - d)$, while at high values of dilution a {\em parallel ordering}, where all these amplitudes collapse to the same value, prevails: the phase diagrams of these networks properly capture these different working regions. 
\newline
Remarkably, confined to the low storage regime (where glassy phenomena can be neglected), the presence (or its lacking) of a teacher does not alter the above scenario and the threshold for a secure learning, namely the minimal required amount of examples (given the constraints, that is the noise in the data-set r, the amount of different archetypes $K$ to cope with, etc.) $M_\otimes$ that guarantees that the network is able to infer the archetype and thus generalize, is the same for supervised and unsupervised protocols and its value has been explicitly calculated: this is another key point toward sustainable AI.
\newline
Clearly there is still a long way to go before a full statistical mechanical theory of extensive parallel processing will be ready, yet this paper acts as a first step in this direction and we plan to report further steps in a near future.

\newpage

\appendix

\section{A more general sampling scenario} \label{app:sampling}
The way in which we add noise over the archetypes to generate the data-set in the main text (see eq. (\ref{eq:quenched_chi})) is a rather peculiar one as, in each example, it preserves the number but also the positions of lacun\ae ~already present in the related archetype. This implies that the noise can not affect the amplitudes of the original signal, i.e. $\sum_i (\eta_i^{\mu,a})^2 = \sum_i (\xi_i^{\mu})^2$ holds for any $a$ and $\mu$, while we do expect that with more general kinds of noise the dilution, this property is not preserved sharply. 
\newline
Here we consider the case where the number of blank entries present in $\boldsymbol \xi^{\mu}$ is preserved on average in the related sample $\{\eta^{\mu,a}\}_{a=1,...,M}$ but lacun\ae ~can move along the examples: this more realistic kind of noise gives rise to cumbersome calculations (still analytically treatable) but should not affect heavily the capabilities of learning, storing and retrieving of these networks (as we now prove). 

Specifically here we define the new kind of examples $\Tilde{\eta}_{i}^{\mu,a}$ (that we can identify from the previous ones $\eta_{i}^{\mu,a}$ by labeling them with a tilde) in the following way 
\begin{definition}
    Given $K$ random patterns $\boldsymbol \xi^{\mu}$ ($\mu=1,...,K$), each of length $N$, whose entries are i.i.d. from
\begin{equation}
\mathbb  P(\xi_i^{\mu}) = \frac{(1-d)}{2}\delta_{\xi_i^{\mu},-1} +  \frac{(1-d)}{2}\delta_{\xi_i^{\mu},+1} + d  \delta_{\xi_i^{\mu},0},   
\end{equation}
we use these archetypes to generate $M\times K$ different examples $\{\tilde\eta_i^{\mu,a}\}^{a=1,\hdots, M}$ whose entries are depicted following
\begin{equation}
\label{eq:new_eta}
\begin{array}{lll}
     \mathbb P(\tilde\eta_i^{\mu,a}|\xi_i^\mu=\pm 1)=A_\pm(r,s)\delta_{\tilde\eta_i^{\mu,a},\xi_i^{\mu}}+ B_\pm(r,s)\delta_{\tilde\eta_i^{\mu,a},-\xi_i^\mu}+ C_\pm(r,s)\delta_{\tilde\eta_i^{\mu,a},0}
\\\\
    \mathbb P(\tilde\eta_i^{\mu,a}|\xi_i^\mu=0)=A_0(r,s)\delta_{\tilde\eta_i^{\mu,a},\xi_i^\mu}+ B_0(r,s)\delta_{\tilde\eta_i^{\mu,a},+1}+ C_0(r,s)\delta_{\tilde\eta_i^{\mu,a},-1}
\end{array}
\end{equation}
for $i=1,\hdots, N$ and $\mu=1,\hdots, K$, where we pose
\begin{equation}\label{Dd}
\begin{array}{lll}
     A_\pm(r,s)=\dfrac{1+r}{2}\left[1-\dfrac{d}{1-d}(1-s)\right]+\dfrac{d(1-s)(1- r)}{4(1-d)}\,,&\;\;&\;\;A_0(r,s)=\dfrac{1+s}{2}\,,
     \\\\
     B_\pm(r,s)= \dfrac{1-r}{2}\left[1-\dfrac{d}{1-d}(1-s)\right]+\dfrac{d(1-s)(1+r)}{4(1-d)}\,,&\;\;&\;\;B_0(r,s)=\dfrac{1-s}{4}\,,
     \\\\
     C_\pm(r,s)=\dfrac{d}{2(1-d)}(1-s)\,,&\;\;&\;\;C_0(r,s)=\dfrac{1-s}{4}\,,
\end{array}
\end{equation}
with $r,s\in[0;1]$ (whose meaning we specify soon, {\em vide infra}). 
\end{definition}
Equation \eqref{eq:new_eta}  codes for the new noise, the values of the coefficients presented in \eqref{Dd} have been chosen in order that all the examples contain  on average the same fraction $d$ of null entries as the original archetypes. To see this it is enough to check that the following relation holds for each $a=1,\hdots, M$, $i=1,\hdots,N$ and $\mu=1,\hdots,K$
\begin{equation}
\begin{array}{lll}
\mathbb P(\tilde\eta_i^{\mu,a}=0)=\SOMMA{x\in\{-1,0,1\}}{}\hspace*{-0.3cm}\mathbb P(\tilde\eta_i^{\mu,a}=0|\xi_i^{\mu}=x)\mathbb P(\xi_i^{\mu}=x)= C_\pm(r,s) (1-d)+ A_0(r,s) d=d\,.
\end{array}
\end{equation}
Defined the data-set, the cost function follows straightforwardly in Hebbian settings as
\begin{definition}
Once introduced $N$ Ising neurons $\sigma_i = \pm 1$ ($i=1,...,N$) and the data-set considered in the definition above, the Cost function of the multitasking Hebbian network equipped with not-preserving-dilution noise reads as
\begin{equation}\label{NewCost1}
\mathcal{H}^{(sup)}_{N,K,M, r,s, d}(\boldsymbol{\sigma} \vert\tilde{\boldsymbol{\eta}}) = -\dfrac{1}{N}\dfrac{1}{ (1-d)(1+\tilde\rho)}\SOMMA{\mu=1}{K}\SOMMA{i,j=1}{N,N}\left(\dfrac{1}{\tilde{r}M }\SOMMA{a=1}{M}\tilde\eta_i^{\mu,a}\right)\left(\dfrac{1}{\tilde{r} M }\SOMMA{b=1}{M}\tilde\eta_j^{\mu,b}\right)\sigma_i\sigma_j,
\end{equation}
where 
\begin{equation}
    \tilde r = \dfrac{r}{(1-d)}\left[1-\dfrac{d}{2} (5-3 s)\right]
\end{equation}
and $\tilde\rho$ is the generalization of the data-set entropy, defined as:
\begin{equation}
    \begin{array}{lll}
        \tilde\rho= \dfrac{1-\tilde{r}^2}{M \tilde{r}^2}\,.
    \end{array}
\end{equation}
\end{definition}


%
\begin{definition}
The suitably re-normalized example's magnetizations $n_{\mu}$ read as
\begin{align}
    n_\mu& :=\dfrac{1}{(1+\tilde\rho)}\dfrac{1}{N}\SOMMA{i=1}{N} \left(\dfrac{1}{\tilde r M}\SOMMA{a=1}{M}\tilde\eta_i^{\mu,a}\right)\sigma_i\,.
\end{align}
\end{definition}
En route toward the statistical pressure, still preserving Guerra's interpolation as the underlying technique, we give the next
\begin{definition}
Once introduced the noise $\beta \in \mathbb{R}^+$, an interpolating parameter $t \in (0,1)$, the $K+1$ auxiliary fields $J$ and $\psi_{\mu}$ ($ \mu \in (1,...,K)$), the interpolating partition function related to the model defined by the Cost function \eqref{NewCost1} reads as
\begin{equation}
\label{NewZeta}
\begin{array}{lll}\small
\mathcal{Z}^{(sup)}_{\beta, N,K,M,r,s,d}( \bm \xi, \tilde{\bm \eta} ;J, t)=\SOMMA{\lbrace\boldsymbol{\sigma}\rbrace}{} \exp{}\Bigg[&J \SOMMA{\mu,i=1}{K,N} \xi_i^\mu \si+t\beta N \dfrac{(1+\rhoD)}{2(1-d)}\SOMMA{\mu=1}{K} n_{\mu}^2(\boldsymbol{\sigma})+(1-t)N\SOMMA{\mu=1}{K}\psi_{\mu} \, n_{\mu}(\boldsymbol{\sigma})\Bigg].
\end{array}
\end{equation}
and the interpolating statistical pressure $\mathcal A_{\beta, K,M, r,s,d}= \lim\limits_{N \to \infty} A_{\beta, N,K,M, r,s,d}$  induced by the partition function (\ref{NewZeta}) reads as
\begin{equation}
 A_{\beta, N,K,M, r,s,d}(J, t) = \frac{1}{N} \mathbb{E}\Big[\ln \mathcal{Z}^{(sup)}_{\beta, N,K,M, r,s,d}( \bm \xi, \tilde{\bm \eta} ;J, t)\Big]
\end{equation}
where $\mathbb{E}=\mathbb{E}_{\xi}\mathbb{E}_{(\tilde\eta|\xi)}$.
\end{definition}
\begin{remark}
Of course, as in the model studied in the main text still with Guerra's interpolation technique, we aim to find an explicit expression (in terms of the control and order parameters of the theory) of the interpolating statistical pressure evaluated at $t=1$ and $J=0$.
\end{remark}
We thus perform the computations following the same steps of the previous investigation: the $t$ derivative of interpolating pressure is given by 
\begin{equation}
\begin{array}{lll}
     \dfrac{d\mathcal A_{\beta,K,M,r,s,d}( J, t)}{dt}=&\dfrac{\beta}{2(1-d)}(1+\rhoD)\SOMMA{\mu=1}{K}\l n^2_\mu \r_t-\SOMMA{\mu=1}{K}\psi_\mu\l n_\mu\r_t.
\end{array}
\end{equation}
fixing 
\begin{equation}
    \psi_\mu =\dfrac{\beta}{1-d}(1+\rhoD)\n_{\mu}
\end{equation}
and computing the one-body term 
\begin{equation}
\begin{array}{lll}
      \mathcal{A}_{\beta, K,M,r,s,d}(J, t=0)&=\mathbb{E}\ln{}\left[{2\cosh\left(\SOMMA{\mu=1}{K}\psi_\mu \dfrac{1}{(1+\rhoD)}\dfrac{1}{\tilde r M}\SOMMA{a=1}{M}\tilde\eta_i^{\mu,a}+J\SOMMA{\mu=1}{K}\xi^\mu\right)}\right]
      \\\\
      &=\mathbb{E}\ln{}\left\{{2\cosh\left[\dfrac{\beta}{1-d}\SOMMA{\mu=1}{K}\n_{\mu} \left(\dfrac{1}{\tilde r M}\SOMMA{a=1}{M}\tilde\eta_i^{\mu,a}\right)+J\SOMMA{\mu=1}{K}\xi^\mu\right]}\right\}.
\end{array}
\end{equation}
We get the final expression as $N \to \infty$ such that we can state the next 
\begin{theorem}
In the thermodynamic limit $(N \to \infty)$ and in the low load regime $(K/N \to 0)$, the quenched statistical pressure of the multitasking Hebbian network equipped with not-preserving-dilution noise, whatever the presence of a teacher,  reads as   	
\begin{equation}
\begin{array}{lll}
\mathcal{A}_{\beta, K,M, r,s,d}(J)
 &=&\mathbb{E}\left\{\ln{}\Bigg[2\cosh\left(\b\SOMMA{\mu=1}{K}\n_{\mu} \hat\eta^\mu+J\SOMMA{\mu=1}{K}\xi^\mu\right)\Bigg]\right\}-\dfrac{\b}{2}(1+\rhoD)\SOMMA{\mu=1}{K}\n_\mu^2.
\end{array}
\end{equation}
\normalsize
where $\b=\beta/(1-d)$, $\mathbb{E}=\mathbb{E}_{\xi}\mathbb{E}_{(\tilde\eta|\xi)}$ and $\hat{\eta}^\mu=\dfrac{1}{\tilde r M}\SOMMA{a=1}{M}\tilde\eta_i^{\mu,a} $and the values $\bar{n}_\mu$  must fulfill the following self-consistent equations
\begin{equation}
    \begin{array}{lll}
        \n_\mu=\dfrac{1}{(1+\rhoD)}\mathbb{E}\left\{\Bigg[\tanh\left(\b\SOMMA{\nu=1}{K}\n_{\nu} \hat\eta^\nu\right)\Bigg]\hat\eta^\mu\right\} &\;\;&\mathrm{for}\;\;\mu=1,\hdots, K\,,
    \end{array}
\end{equation}
that extremize the statistical pressure $\mathcal{A}_{\beta, K,M, r,s,d}(J=0)$ w.r.t. them.
\end{theorem}

Furthermore, the simplest path to obtain a self-consistent equation also for the Mattis magnetization $m_{\mu}$ is by considering the auxiliary field $J$ coupled to $m_{\mu}$ namely $\bar{m}_\mu=  \nabla_J \mathcal{A}_{\beta, K,M, r,s,d}(J)\vert_{J=0}$ to get
\begin{equation}
    \begin{array}{lll}
    \m_\mu = \mathbb{E} \left\{\tanh\left[\b\SOMMA{\nu=1}{K}\n_{\nu} \hat\eta^\nu \right] \xi^\mu\right\}&\;\;&\mathrm{for}\;\;\mu=1,\hdots, K\,.
    \end{array}
\end{equation}
 \begin{figure}
     \centering
     \includegraphics[scale=1.2]{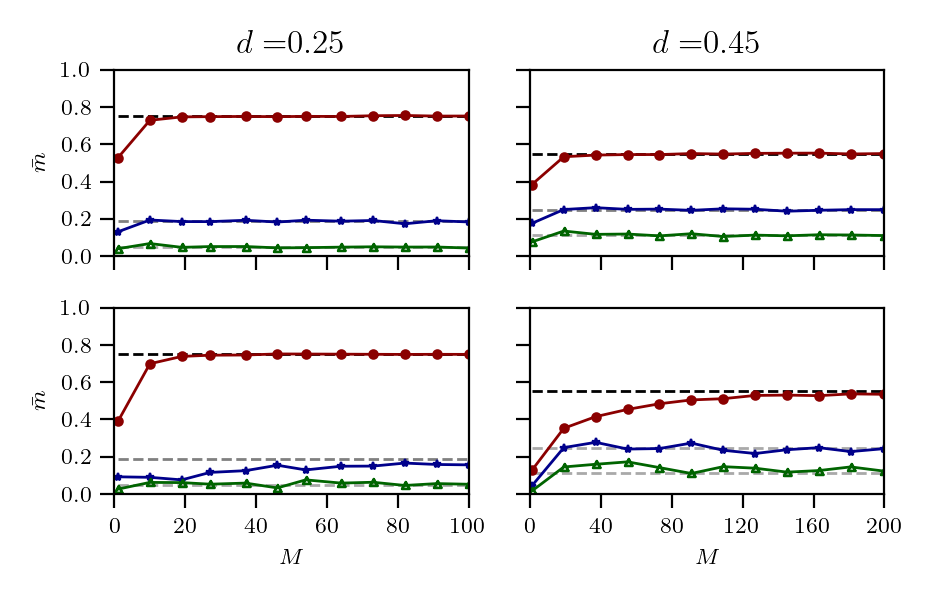}
     \caption{Comparison  of the numerical solution of the self consistency equations related to the Mattis magnetization in the two models: upper panel is due to the first model (reported in the main text), lower panel reports on the second model (deepened here). Beyond a different transient at small $M$ the two models behave essentially in the same way.}
 \end{figure}
We do not plot these new self-consistency equations as, in the large $M$ limit, there are no differences w.r.t. those obtained in the main text.






\section{On the data-set entropy $\rho$}\label{AppEntropy}

In this appendix, focusing on a single generic bit, we  deepen the relation between the conditional entropy  $H(\xi^\mu_i|\bm \eta^{\mu}_i)$ of a given pixel $i$ regarding archetype $\mu$ and the information provided by the data-set regarding such a pixel, namely the block $ \left(\eta^{\mu,1}_i, \eta^{\mu,2}_i, \hdots, \eta^{\mu,M}_i\right)$ to justify why we called $\rho$ the data-set entropy in the main text. As calculations are slightly different among the two analyzed models (the one preserving dilution position provided in the main text and the generalized one given in the previous appendix) we repeat them model by model  for the sake of transparency.

\subsection{I: multitasking Hebbian network equipped with not-affecting-dilution noise}
Let us focus on the $\mu$-th pattern and the $i$-th digit, whose related block is 
\begin{equation}
    \eta^{\mu}_i = \left(\eta^{\mu,1}_i, \eta^{\mu,2}_i, \hdots, \eta^{\mu,M}_i\right);
\end{equation} 
the error probability for any single entry is 
\begin{equation}
    \mathbb{P}(\xi^{\mu}_i \neq 0)\mathbb{P}(\eta^{\mu,a}_i \neq \xi_i^\mu) = (1-d)(1 - r)/2
\end{equation}
 and, by applying the majority rule on the block, it is reduced to
 \begin{equation}
     \mathbb{P}(\xi^{\mu}_i \neq 0)\mathbb P\left(\mathrm{sign}\Big(\sum\limits_a \eta^{\mu,a}_i\xi_i^\mu\Big) = -1\right) \underset{M\gg 1}{\approx}\dfrac{(1-d)}{2}\left[1 -\mathrm{erf}\left(\dfrac{1}{\sqrt{2\rho}}\right)\right].
 \end{equation}
 Thus
 \begin{equation}
     H_{d,r,M}(\bm\xi^\mu|\bm \eta^{\mu}) = -\left[x(d,r,M)\log_2 x(d,r,M)+y(d,r,M)\log_2 y(d,r,M)\right]
     \label{eq:entropy_standard}
 \end{equation}
 where
 \begin{equation}
 \begin{array}{lll}
      x(d,r,M) = \dfrac{(1-d)}{2}\left[1 -\mathrm{erf}\left(\dfrac{1}{\sqrt{2\rho}}\right)\right]\;,
      &&
      y(d,r,M) = 1- x(d,r,M)\,.
 \end{array}
 \end{equation}
 

\subsection{II: multitasking Hebbian network equipped with not-preserving-dilution noise }

Let us focus on the $\mu$-th pattern and the $i$-th digit, whose related block is 
\begin{equation}
    \tilde\eta^{\mu}_i = \left(\tilde\eta^{\mu,1}_i, \tilde\eta^{\mu,2}_i, \hdots, \tilde\eta^{\mu,M}_i\right);
\end{equation} 
the error probability for any single entry is 
\begin{equation}
    \mathbb{P}(\xi^{\mu}_i \neq 0)\mathbb{P}(\tilde\eta^{\mu,a}_i \xi^{\mu}_i \neq +1|\xi^{\mu}_i \neq 0)+\mathbb{P}(\xi_i^\mu=0)\mathbb{P}(\tilde\eta^{\mu,a}_i \neq 0|\xi_i^\mu=0) = d(1-s)\,.
\end{equation}
By applying the majority rule on the block, it is reduced to
 \begin{equation}
 \begin{array}{lll}
       \mathbb{P}(\xi^{\mu}_i \neq 0)\left[1-\mathbb{P}\Big(\mathrm{sign}(\hat\eta^{\mu}_i\xi_i^\mu) = +1\Big|\xi^{\mu}_i \neq 0\Big)\right]+\mathbb{P}(\xi^{\mu}_i = 0)\mathbb{P}\Big(\mathrm{sign}[|\hat\eta^{\mu}_i|] = +1\Big|\xi_i^\mu=0\Big)
       \\\\
       \underset{M\gg 1}{\approx} \dfrac{(1-d)}{2}\left\{1-\mathrm{erf}\left[\left(2\tilde\rho-\dfrac{d (1-s)}{(1- d)M \tilde r^2}\right)^{-1/2}\right]\right\}+\dfrac{d}{2}\;\left\{1-\mathrm{erf}\left[\left(\dfrac{1-s}{M\tilde r^2}\right)^{-1/2}\right]\right\}\,. 
 \end{array}
 \end{equation}
Thus
 \begin{equation}
H_{d,r,s,M}(\xi^\mu_i|\tilde{\bm \eta}^{\mu}_i) = -\left[x(d,r,s,M)\log_2 x(d,r,s,M)+y(d,r,s,M)\log _2 y(d,r,s,M)\right]
     \label{eq:entropy_standard}
 \end{equation}
 where
 \begin{equation}
 \begin{array}{lll}
      x(d,r,s,M) = \dfrac{(1-d)}{2}\left\{1-\mathrm{erf}\left[\left(2\tilde\rho-\dfrac{d (1-s)}{(1- d)M \tilde r^2}\right)^{-1/2}\right]\right\}+\dfrac{d}{2}\;\left\{1-\mathrm{erf}\left[\left(\dfrac{1-s}{M\tilde r^2}\right)^{-1/2}\right]\right\}
      \\\\
      y(d,r,s,M) = 1-x(d,\rhoD)
 \end{array}
 \end{equation}
Whatever the model, the conditional entropies $H_{d,r,M}(\xi^\mu_i|\bm \eta^{\mu}_i)$ and $H_{d,r,s,M}(\xi^\mu_i|\tilde{\bm \eta}^{\mu}_i)$ are a monotonic increasing functions of $\rho$ and $\tilde \rho$ respectively, hence the reason to calling $\rho$ and $\tilde \rho$ the entropy of the data-set.

\section{Stability analysis: an alternative approach} \label{AppendiceC}
\subsection{Stability analysis via signal-to-noise technique}
\label{subsec:S2N_sup}
The standard signal-to-noise technique \cite{Amit} is a powerful method to investigate the stability of a given neural configuration in the noiseless limit $\beta \to \infty$: by requiring that each neuron is aligned to its field (the post-synaptic potential that it is experiencing, i.e. $h_i \sigma_i \geq 0 \ \ \forall i \in (1,...,N)$) this analysis allows to correctly classify which solution (stemming from the self-consistent equations for the order parameters) is preferred as the control parameters are made to vary and thus it can play as an alternative route w.r.t the standard study of the Hessian of the statistical pressure reported in the main text (see Sec. \ref{cazzarola}).
\newline
In particular, recently, a revised version of the signal-to-noise technique has been developed \cite{Pedreschi1,Pedreschi2} and in this new formulation it is possible to obtain the self-consistency equations for the order parameters explicitly so we can compare directly outcomes from signal-to-noise to outcomes from statistical mechanics. By comparison of the two routes that lead to the same picture, that is statistical mechanics and revised signal-to-noise technique, we can better comprehend the working criteria of these neural networks.
\newline
\newline
We  suppose that the network is in the  hierarchical configuration prescribed by eq.  (\ref{eq:parallel_ans}), that we denote as $\bm \sigma = \bm \sigma^*$, and we must evaluate the local field $h_i(\bm \sigma^*)$ acting on the generic neuron $\sigma_i$ in this configuration to check that $h_i(\bm \sigma^*)\si^* >0$ is satisfied for any $i=1, \hdots, N$: should this be the case, the configuration would be stable, vice versa unstable.
\newline
Focusing on the supervised setting with no loss of generality (as we already discussed that the teacher essentially plays no role in the low storage regime) and selecting (arbitrarily) the hierarchical ordering as a test case to be studied, we start by re-writing the Hamiltonian \eqref{def:H_dil_sup} as 
\begin{equation}
    -\mathcal{H}_{N,K,M, r}(\boldsymbol{\sigma} \vert \bm \eta)=\SOMMA{i=1}{N}h_i(\boldsymbol \sigma)\sigma_i \, ,
\end{equation}
where the local fields $h_i$ appear explicitly and are given by 
\begin{align}
    h_i(\boldsymbol{\sigma})=& \dfrac{1}{2N\,r^2M^2 (1-d)(1+\rho)}\SOMMA{\mu=1}{K}\SOMMA{j\neq i}{N}\SOMMA{a,b}{M_\mu,M_\mu}\eta^{\mu, a}_{i}\eta^{\mu,b}_{j}\sigma_{j} \, .
\end{align}
The updating rule for the neural dynamics reads as
\begin{equation}    \sigma^{(n+1)}_i=\sigma^{(n)}_i\mathrm{sign}\left(\tanh{\left[\beta\sigma_i^{(n)}h_i^{(n)}(\boldsymbol{\sigma}^{(n)})\right]}+\Gamma_i\right) \;\;\mathrm{with}\;\;\Gamma_i\sim \mathcal{U}[-1;+1] \, ,
\end{equation}
that, in the zero fast-noise limit $\beta\to +\infty$, reduces to
\begin{equation}    \sigma^{(n+1)}_i=\sigma^{(n)}_i\mathrm{sign}\left(\sigma_i^{(n)}h_i^{(n)}(\boldsymbol{\sigma}^{(n)})\right).
\end{equation}
To inspect the stability of the hierarchical parallel configuration, we initialize the network in such configuration, i.e., $\boldsymbol \sigma^{(1)} = \bm \sigma^*$, then, following Hinton's prescription \cite{Hinton1MC,Hinton2MC}\footnote{The {\em early stopping prescription} given by Hinton and coworkers became soon very popular, yet it has been criticized in some circumstances (in particular where glassy features are expected to be strong and may keep the network out of equilibrium for very long times, see e.g. \cite{McKay-2001,AuroBea1,AuroBea2}): we stress that, in the present section, we are assuming the network has already reached the equilibrium further, confined to the low storage inspection, spin glass bottlenecks in thermalization should not be prohibitive.} the one-step n-iteration $\bm \sigma^{(2)}$ leads to an expression of the magnetization that reads as
\begin{equation}
    m_\mu^{(2)}=\dfrac{1}{N}\SOMMA{i=1}{N}\xi_i^\mu\sigma_i^{(2)}=\dfrac{1}{N}\SOMMA{i=1}{N}\xi_i^\mu\left[\sigma^*_i\mathrm{sign}\left(\sigma_i^*\,h_i^{(1)}(\boldsymbol{\sigma}^{*})\right)\right];
    \label{eq:m_1_step}
\end{equation}
Next, using the explicit expression of the hierarchical parallel configuration \eqref{eq:parallel_ans}, we get
\begin{equation}
\small
\begin{array}{lll}
     m_1^{(2)}=\dfrac{1}{N}\SOMMA{i=1}{N}\left(\xi_i^1\right)^2\mathrm{sign}\left(\sigma_i^*\,h_i^{(1)}(\boldsymbol{\sigma}^{*})\right);
     \\\\
     m_{\mu>1}^{(2)}=\dfrac{1}{N}\SOMMA{i=1}{N}\left(\xi_i^\mu\right)^2\prod\limits_{\rho=1}^{\mu-1}\delta\left(\xi_i^{\rho}\right)\mathrm{sign}\left(\sigma_i^*\,h_i^{(1)}(\boldsymbol{\sigma}^{*})\right);
\end{array}
\end{equation}
by applying the central limit theorem to estimate the sums appearing in the definition of $h_i^{(1)}$ for $i=1, \hdots, N$, we are able to split, mimicking the standard signal-to-noise technique, a signal contribution $(\kappa_{1,\mu}^{(1)})$ and a noise contribution $(\kappa_{2,\mu}^{(1)})$ as presented in the following 
\begin{equation}
\begin{array}{lll}
     \sigma^*_i h_i^{(1)} (\boldsymbol{\sigma}^{*})\sim \kappa_{1,\mu}^{(1)} + z_i \sqrt{\kappa_{2,\mu}^{(1)}}\;\;\;\;\;\;\mathrm{with}\;\;z_i\sim\mathcal{N}(0,1)
\end{array}
\end{equation}
where
\begin{eqnarray}
\kappa_{1,\mu}^{(1)}\coloneqq \mathbb{E}_{\xi}\mathbb{E}_{(\eta|\xi)} \left[\sigma^*_i\,h_i^{(1)}(\boldsymbol{\sigma}^{*})\right] &\;\;\;\; \kappa_{2,\mu}^{(1)}\coloneqq \mathbb{E}_{\xi}\mathbb{E}_{(\eta|\xi)} \left[\sigma^*_i\,h_i^{(1)}(\boldsymbol{\sigma}^{*})\right] ^2
\end{eqnarray}
Thus, Eq.~\eqref{eq:m_1_step} becomes
\begin{equation}
\begin{array}{lll}
        m_\mu^{(2)}&=&\left[\dfrac{1}{N}\SOMMA{i=1}{N}\left(\xi_i^\mu\right)^2\prod\limits_{\rho=1}^{\mu-1}\delta\left(\xi_i^{\rho}\right)\mathrm{sign}\left(\kappa_{1,\mu}^{(1)} +z_i\sqrt{\kappa_{2,\mu}^{(1)}}\right)\right]\;\;\;\;\;\mathrm{with}\;\;\mu=1,2,\hdots,K\,.
\end{array}
    \label{eq:m_1_step_sign_1}
\end{equation}
\begin{figure}
    \centering
    \includegraphics[width=15cm]{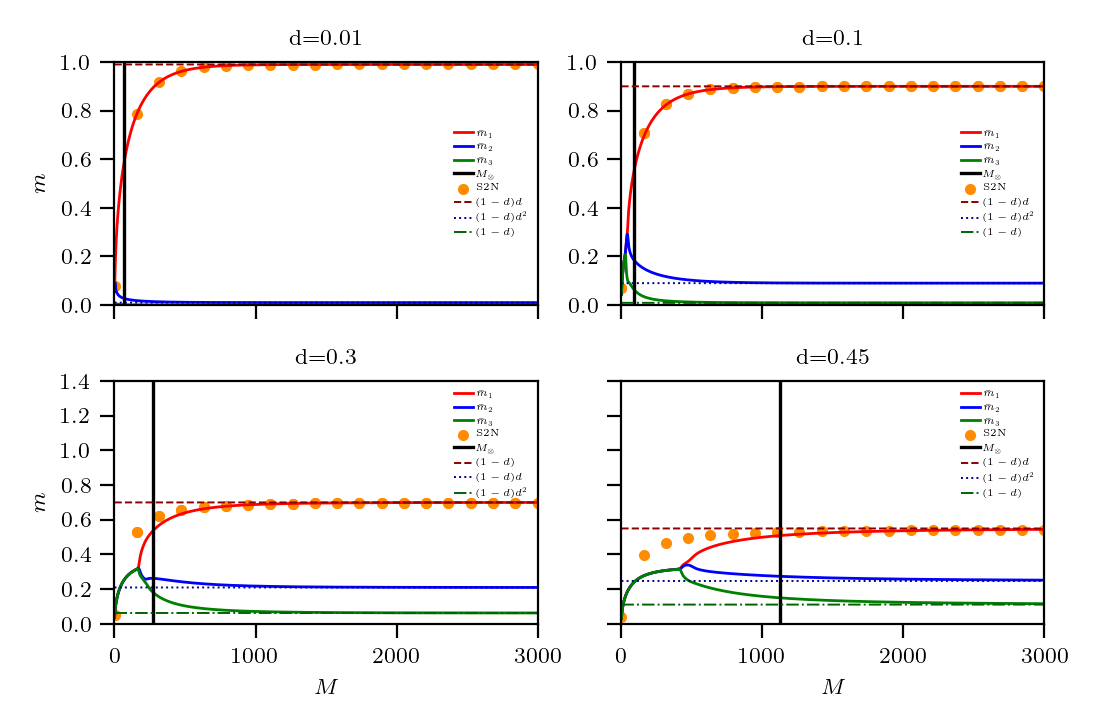}
    \caption{Signal to Noise numerical inspection of the Mattis magnetizations for a diluted network with $r=0.1$ and $K=3$ in the hierarchical regime (at  levels of pattern's dilution $d<d_c$ as reported in their titles): we highlight the agreement, as the saturation level is reached, among Signal to Noise analysis (orange dots) and the value of the magnetization of the first pattern found by the mechanical statistical approach (reported as solid red line). The dashed lines represent the Hebbian storing prescriptions at which the values of the magnetizations converge. The vertical black line depicts the critical amount of examples $M_\otimes$ that must be experienced by the network to properly depict the archetypes: note that this value is systematically above, in M, that the point where all the bifurcations happened, hence all the magnetizations stabilized on their hierarchical displacements.}
    \label{fig:s2n_3}
\end{figure}
For large values of $N$, the arithmetic mean coincides with the theoretical expectation, thus
\begin{equation}
    \dfrac{1}{N}\SOMMA{i=1}{N} g(\xi_i,z_i) \;\; \xrightarrow[N\to\infty]{} \mathbb{E}_{\xi,z}[g(\xi,z)]=\mathbb{E}_{\xi}\left[\displaystyle\int \dfrac{dz}{\sqrt{2\pi}}e^{-\frac{z^2}{2}} g(\xi,z)\right]\,.
\end{equation}
therefore, we can rewrite Eq. \eqref{eq:m_1_step_sign_1} as
\begin{equation}
\begin{array}{lll}
        m_\mu^{(2)}&=& (1-d)d^{\mu-1}\mathrm{erf}\left[\dfrac{\kappa_{1,\mu}^{(1)}}{\sqrt{2\Big(\kappa_{2,\mu}^{(1)}-\left(\kappa_{1,\mu}^{(1)}\right)^2\Big)}}\right]\;\;\;\;\;\mathrm{with}\;\;\mu=1,2,\hdots,K\,.
\end{array}
\end{equation}

While we carry out the computations of $\kappa_{1,\mu}^{(1)}$ and $\kappa_{2,\mu}^{(1)}$ in Appendix \ref{app:momenta}, here we report only their values, which are
\begin{equation}
\begin{array}{lll}
     \kappa_{1,\mu}^{(1)}= \dfrac{1}{2}\dfrac{1}{(1+\rho)}d^{\mu-1}\,,\;\;
    &&
     \kappa_{2,\mu}^{(1)}= \dfrac{1}{4}\dfrac{1}{1+\rho}d^{\mu-1}\left[\dfrac{1+d+d^2-d^{2K-2\mu+2}}{1+d}\right].
\end{array}
\end{equation}
So we get
\begin{equation}
\begin{array}{lll}
        m_\mu^{(2)} (d,K,\rho)       &=&(1-d)d^{\mu-1}\mathrm{erf}\left[\dfrac{1}{\sqrt{2}}\dfrac{\sqrt{(1+d)d^{\mu-1}}}{\sqrt{(1+\rho)(1+d+d^2-d^{2K-2\mu+2})-d^\mu-d^{\mu-1}}}\right]
\end{array}
\label{eq:m_1_step_sign}
\end{equation}
As shown in Fig.~\ref{fig:s2n_3}, as a critical amount of perceived examples is collected, this expression is in very good agreement with the estimate stemming from the numerical solution of the self-consistent equations and indeed we can finally state the last 
\begin{theorem}
In the zero fast-noise limit ($\beta\to+\infty$), if the neural configuration
\begin{equation}
    \tilde{\bm\sigma}=\tilde{\bm\sigma}(\bm\xi)
\end{equation}
is a fixed point of the dynamics described by the sequential spin upload rule 
\begin{equation} \sigma_i^{(n+1)}=\sigma^{(n)}_i\mathrm{sign}\left[\beta\sigma_i^{(n)}h_i^{(n)}(\boldsymbol{\sigma}^{(n)})\right] 
\end{equation}
where 
\begin{align}
    h_i^{(n)}(\boldsymbol{\sigma})=& \dfrac{1}{N\,M^2(1+\rho)r^2}\SOMMA{\mu=1}{K}\SOMMA{j\neq i}{N}\SOMMA{a,b}{M_\mu,M_\mu}\eta^{\mu, a}_{i}\eta^{\mu,b}_{j}\sigma_{j}^{(n)}\, ,
\end{align}
then the order parameters $n_\mu(\bm\sigma)=[NM(1+\rho)r]^{-1}\sum\limits_i^N\sum\limits_a^M\eta_i^{\mu,a}\sigma_i$ must satisfy the following self equations
\begin{equation}    
\begin{array}{lll}
     n_{\mu}&=&\dfrac{1}{(1+\rho)}\mathbb{E}_{\bm\xi}\mathbb{E}_{(\bm\eta|\bm\xi)}\Bigg\{\hat\eta^\mu\tilde{\sigma}(\bm\xi)\,\mathrm{sign}\left[\SOMMA{\nu=1}{K}n_{\nu}\hat\eta^\nu\tilde{\bm\sigma}(\bm\xi)\right]\Bigg\}\,,
\end{array}
\label{eq:self_theorem}
\end{equation}
where we set $\hat\eta^\mu=(Mr)^{-1}\sum\limits_a^M\eta^{\mu,a}$.
\end{theorem}
\begin{remark}
The empirical evidence that, via early stopping criterion, we still obtain the correct solution proves a posteriori the validity of Hinton's recipe in the present setting and it tacitly candidate statistical mechanics as a reference also to inspect computational shortcuts.
\end{remark}

\begin{proof}
The local fields $h_i$ can be rewrite using the definition of $n_\mu$ as
\begin{align}
    h_i^{(n)}(\boldsymbol{\sigma})=\SOMMA{\mu=1}{K}n_{\mu}^{(n)}(\bm\sigma)\hat\eta^{\mu}_{i} \, ,
\end{align}
in this way the upload rule can be recast as
\begin{equation}    
\sigma^{(n+1)}_i=\sigma^{(n)}_i\mathrm{sign}\left[\SOMMA{\mu=1}{K}n_{\mu}^{(n)}(\bm \sigma)\hat\eta^{\mu}_{i}\sigma_i^{(n)}\right].
\end{equation}
Computing the value of the $n_\mu$-order parameters at the $(n+1)$ step of uploading process we get
\begin{equation}    
\begin{array}{lll}
     n_{\mu}^{(n+1)}(\bm \sigma)&=&\dfrac{1}{N(1+\rho)}\SOMMA{i}{N}\hat\eta^{\mu}_{i}\sigma^{(n+1)}_i
     \\\\
     &=&\dfrac{1}{N(1+\rho)}\SOMMA{i}{N}\hat\eta^{\mu}_{i}\sigma^{(n)}_i\mathrm{sign}\left[\:\SOMMA{\nu=1}{K}n_{\nu}^{(n)}(\bm \sigma)\hat\eta^{\nu}_{i} \sigma_i^{(n)}\right]\,.
\end{array}
\label{eq:semi_self}
\end{equation}
If $\tilde{\bm\sigma}(\bm\xi)$ is a fixed point of our dynamics, we must have
$\tilde{\bm\sigma}^{(n+1)}\equiv \tilde{\bm\sigma}^{(n)}$ and $n^{(n+1)}(\bm\sigma)\equiv n^{(n)}(\bm\sigma)$, thus \eqref{eq:semi_self} becomes
\begin{equation}    
\begin{array}{lll}
     n_{\mu}(\bm \sigma)&=&\dfrac{1}{N(1+\rho)}\SOMMA{i}{N}\hat\eta^{\mu}_{i}\tilde\sigma_i(\bm\xi)\,\mathrm{sign}\left[\SOMMA{\nu=1}{K}n_{\nu}(\bm \sigma)\hat\eta^{\nu}_{i} \tilde\sigma_i(\bm\xi)\right]. 
\end{array}
\label{eq:semi_self2}
\end{equation}
For large value of $N$, the arithmetic mean coincides with the theoretical expectation, thus
\begin{equation}
    \dfrac{1}{N}\SOMMA{i=1}{N}g(\eta_i) \xrightarrow[N\to\infty]{}\mathbb{E}_{\eta}\Big[g(\eta)\Big]
\end{equation}
therefore, \eqref{eq:semi_self2} reads as
\begin{equation}    
\begin{array}{lll}
     n_{\mu}&=&\dfrac{1}{(1+\rho)}\mathbb{E}_{\bm\xi}\mathbb{E}_{(\bm\eta|\bm\xi)}\left\{\hat\eta^{\mu}\tilde\sigma(\bm\xi)\,\mathrm{sign}\left[\:\SOMMA{\nu=1}{K}n_{\nu}\hat\eta^{\nu}\tilde\sigma(\bm\xi)\right]\right\}\,.
\end{array}
\label{eq:semi_self_finale}
\end{equation}
where we used $\mathbb{E}_{\bm\eta}=\mathbb{E}_{\bm\xi}\mathbb{E}_{(\bm\eta|\bm\xi)}$.
\end{proof}

\begin{corollary}
Under the hypothesis of the previous theorem, if the neural configuration coincides with the parallel configuration
\begin{equation}
    \tilde{\bm\sigma}(\bm\xi)=\bm\sigma^* = \bm \xi^1 +\SOMMA{\nu=2}{K}\bm \xi^\nu \prod\limits_{\rho=1}^{\nu-1}\delta\left(\bm \xi^\rho\right)
\end{equation}
then the order parameters $n_\mu(\bm\sigma)$ must satisfy the following self equation
\begin{equation}
\small
\begin{array}{lll}
     n_{\mu}&=&\dfrac{1}{(1+\rho)}\mathbb{E}_{\bm\xi}\mathbb{E}_{(\bm\eta|\bm\xi)}\Bigg\{\hat\eta^\mu\left(\xi^1 +\SOMMA{\lambda=2}{K}\xi^{\lambda} \prod\limits_{\rho=1}^{\lambda-1}\delta\left( \xi^\rho\right)\right)\,\mathrm{sign}\left[\SOMMA{\nu=1}{K}n_{\nu}\hat\eta^\nu\left(\xi^1 +\SOMMA{\lambda=2}{K}\xi^{\lambda} \prod\limits_{\rho=1}^{\lambda-1}\delta\left( \xi^\rho\right)\right)\right]\Bigg\}\,.
\end{array}
\label{eq:self_finale_star}
\end{equation}
\end{corollary}

\begin{proof}
We only have to replace in \eqref{eq:self_theorem} the explicit form of $\bm\sigma^*$ and we get the proof.
\end{proof}

\subsection{Evaluation of momenta of the effective post-synaptic potential}
\label{app:momenta}

In this section we want to describe the computation of first and second momenta $\kappa_{1,\mu}^{(1)}$ and $\kappa_{2,\mu}^{(1)}$ in Sec. \ref{subsec:S2N_sup}, we will present only tha case of $\mu=1$.

Let us start from $\kappa_{1,\mu}^{(1)}$:
\begin{eqnarray}
\kappa_{1,\mu}^{(1)}\Big|_{\xi_i^1=\pm 1}\coloneqq \mathbb{E}_{\xi}\mathbb{E}_{(\eta|\xi)}\left[\sigma_i^{*}h_i^{(1)}(\boldsymbol{\sigma}^{*})\Big|_{\xi_i^1=\pm 1}\right]=\dfrac{1}{2}\SOMMA{j\neq i}{N}\dfrac{1}{r_1 M_1^2(1+\rho)}\mathbb{E}_{\xi}\mathbb{E}_{(\eta|\xi)}\left[\SOMMA{a,b}{M_1,M_1}\eta^{1}_{j}\left(\xi_j^1+\SOMMA{\nu=2}{\hat K}\xi_j^\nu\prod\limits_{\rho=1}^{\nu-1}\delta_{\xi_j^{\rho},0}\right)\right];\notag
\end{eqnarray}
\normalsize
since $\mathbb{E}_{\xi}[\xi_i^{\mu}]=0$ the only non-zero terms are the ones with $\mu=1$:
\small
\begin{eqnarray}
\kappa_{1,\mu}^{(1)}\Big|_{\xi_i^1=\pm 1}&\coloneqq& \dfrac{1}{2}\SOMMA{j\neq i}{N}\dfrac{1}{r_1 M_1^2(1+\rho)}\mathbb{E}_{\xi}\left[\SOMMA{a,b}{M_1,M_1}r_1\xi^{1}_{j}\left(\xi_j^1+\SOMMA{\nu=2}{\hat K}\xi_j^\nu\prod\limits_{\rho=1}^{\nu-1}\delta_{\xi_j^{\rho},0}\right)\right]
\\
&=&\dfrac{1}{2 N M_1^2r_1(1+\rho)}\SOMMA{j\neq 1}{N}\SOMMA{a,b}{M_1,M_1}r_1=\dfrac{1}{2 (1+\rho)}\notag
\end{eqnarray}
\normalsize
where we used $\mathbb{E}_{(\eta|\xi)}[\eta_i^{\mu,a}]=r\xi_i^\mu$. Moving on, we start the computation of $\kappa_{2,\mu}^{(1)}$, due to $\mathbb{E}_{\xi}[\xi_{i_1}^{\mu}\xi_{i_1}^{\nu}]=\delta^{\mu\nu}$, the only non-zero terms are:
\begin{align}
\footnotesize
\begin{array}{llll}
     \kappa_{2,\mu}^{(1)}\Big|_{\xi_i^1=\pm 1}&\coloneqq& \mathbb{E}_{\xi}\mathbb{E}_{(\eta|\xi)}\left[\lbrace \sigma_i^*h_i^{(1)}(\boldsymbol{\sigma}^{*})\rbrace ^2\Big|_{\xi_i^1=\pm 1}\right]
     \\
     &&\dfrac{1}{4N^2 (1-d)^2}\SOMMA{\mu=1}{K}\mathbb{E}_{\xi}\mathbb{E}_{(\eta|\xi)}\dfrac{1}{M_\mu^4r_\mu^4(1+\rho)^2}\SOMMA{k,j\neq i}{N,N}\left(\SOMMA{a_1,a_2,b_1,b_2}{M_\mu}\eta^{\mu, a_1}_{i}\eta^{\mu,a_2}_{i}\eta^{\mu, b_1}_{j}\eta^{\mu,b_2}_{k}\right)
\\     
&&\left(\xi_j^1+\SOMMA{\nu_1=2}{K}\xi_j^{\nu_1}\prod\limits_{\rho_1=1}^{\nu_1-1}\delta\left(\xi_j^{\rho_1}\right)\right)\left(\xi_k^1+\SOMMA{\nu_2=2}{K}\xi_k^{\nu_2}\prod\limits_{\rho_2=1}^{\nu_2-1}\delta\left(\xi_k^{\rho_2}\right)\right) = A_{\mu=1}+B_{\mu>1}
\end{array}
\end{align}
namely we will analyze separately the case for $\mu=1$ ($A_{\mu=1}$) and $\mu>1$ ($B_{\mu>1}$).
\begin{align}
\footnotesize
\begin{array}{llll}
    A_{\mu=1}&=&\dfrac{1}{4N^2 (1-d)^2}\mathbb{E}_{\xi}\mathbb{E}_{(\eta|\xi)}\dfrac{1}{M_1^4r_1^4(1+\rho)^2}\SOMMA{k,j\neq i}{N,N}\left(\SOMMA{a_1,a_2,b_1,b_2}{M_1}\eta^{1, a_1}_{i}\eta^{1,a_2}_{i}\eta^{1, b_1}_{j}\eta^{1,b_2}_{k}\right)\left(\xi_j^1\xi_k^1\right)
    \\\\
    &=&\dfrac{1}{4N^2(1-d)^2}\dfrac{1}{M_1^4r_1^4(1+\rho)^2}\mathbb{E}_{(\eta|\xi=\pm 1)}\left[\SOMMA{a_1,a_2}{M_1}\eta^{1, a_1}_{i}\eta^{1,a_2}_{i}\right]\SOMMA{k,j\neq i}{N,N}\mathbb{E}_{\xi}\mathbb{E}_{(\eta|\xi)}\left[\SOMMA{b_1,b_2}{M_1}\eta^{1, b_1}_{j}\eta^{1,b_2}_{k}\right]\left(\xi_j^1\xi_k^1\right)
    \\\\
    &=&\dfrac{1}{4}\dfrac{1}{(1+\rho)}\,;
\end{array}
\label{eq:amu}
\end{align}

\begin{align}
\footnotesize
\begin{array}{llll}
     B_{\mu>1}&\coloneqq& \dfrac{1}{4N^2 (1-d)^2}\SOMMA{\mu>1}{K}\mathbb{E}_{\xi}\mathbb{E}_{(\eta|\xi)}\dfrac{1}{M_\mu^4r_\mu^4(1+\rho)^2}\SOMMA{k,j\neq i}{N,N}\left(\SOMMA{a_1,a_2,b_1,b_2}{M_\mu}\eta^{\mu, a_1}_{i}\eta^{\mu,a_2}_{i}\eta^{\mu, b_1}_{j}\eta^{\mu,b_2}_{k}\right)
\\\\
&&\left(\prod\limits_{\rho_1=1}^{\mu-1}\delta\left(\xi_j^{\rho_1}\right)\right)\left(\prod\limits_{\rho_2=1}^{\mu-1}\delta\left(\xi_k^{\rho_2}\right)\right)
\\\\
&=& \dfrac{(1-d)^{-2}}{4N^2}\SOMMA{\mu=2}{K}\dfrac{1}{M_\mu^4r_\mu^4(1+\rho)^2}\mathbb{E}_{(\eta|\xi=\pm 1)}\left[\SOMMA{a_1,a_2}{M_\mu}\eta^{\mu, a_1}_{i}\eta^{\mu,a_2}_{i}\right]\SOMMA{k,j\neq i}{N,N}\mathbb{E}_{\xi}\mathbb{E}_{(\eta|\xi)}\left[\SOMMA{b_1,b_2}{M_\mu}\eta^{\mu, b_1}_{j}\eta^{\mu,b_2}_{k}\right] \left(\prod\limits_{\rho_1=1}^{\mu-1}\delta\left(\xi_j^{\rho_1}\right)\right)\left(\prod\limits_{\rho_2=1}^{\mu-1}\delta\left(\xi_k^{\rho_2}\right)\right)
\\\\
&=& \dfrac{(1-d)^{-2}}{4N^2}\SOMMA{\mu=2}{K}\dfrac{1}{M_\mu^2(1+\rho)}\SOMMA{k,j\neq i}{N,N}\SOMMA{b_1,b_2}{M_\mu}\mathbb{E}_{\xi}\left[\xi^{\mu}_{j}\xi^{\mu}_{k} \left(\prod\limits_{\rho_1=1}^{\mu-1}\delta\left(\xi_j^{\rho_1}\right)\right)\left(\prod\limits_{\rho_2=1}^{\mu-1}\delta\left(\xi_k^{\rho_2}\right)\right)\right]
\\\\
&=& \dfrac{1-d}{4(1+\rho)}\SOMMA{\mu=2}{K}d^{2(\mu-1)}= \dfrac{1-d}{4(1+\rho)}\dfrac{d^2-d^{2K}}{1-d^2}=\dfrac{1}{4(1+\rho)}\dfrac{d^2-d^{2K}}{1+d}\,.
\end{array}
\label{eq:bmu}
\end{align}
Putting together Eq. \eqref{eq:amu} and Eq. \eqref{eq:bmu} we reach the expression of $\kappa_{2,\mu}^{(1)}\Big|_{\xi_i^1=\pm 1}$.

\section{Explicit Calculations and Figures for the cases $K=2$ and $K=3$}
\label{app:calc}
In this appendix, we collect the  explicit expression of the self-consistent equations in \eqref{eq:n_Mgrande} and \eqref{eq:m_Mgrande} (focusing only on the cases of $K=2$ and $K=3$) and some Figures obtained from their numerical solution.
\subsection{$K=2$}
Fixing $K=2$ and explicitly perform the mean with respect to $\xi$, \eqref{eq:n_Mgrande} and \eqref{eq:m_Mgrande} read as
\begin{equation}
\footnotesize
    \begin{array}{lll}
    \n_1&=&\dfrac{\m_1}{(1+\rho)}+\dfrac{\b(1-d)\rho\,\n_1}{(1+\rho)}\Bigg[1-d\,\mathcal{S}_2(\n_1,0)-\dfrac{(1-d)}{2}\mathcal{S}_2(\n_1,-\n_2)-\dfrac{(1-d)}{2}\mathcal{S}_2(\n_1,\n_2)\Bigg]
    \\\\
    \n_2&=&\dfrac{\m_2}{(1+\rho)}+\dfrac{\b(1-d)\rho\,\n_2}{(1+\rho)}\Bigg[1-d\,\mathcal{S}_2(0,\n_2) -\dfrac{(1-d)}{2}\mathcal{S}_2(\n_1,-\n_2)-\dfrac{(1-d)}{2}\mathcal{S}_2(\n_1,\n_2)\Bigg]
    \\\\
     \m_1 &=& \dfrac{(1-d)^2}{2}\Bigg[ \mathcal{T}_2(\n_1,\n_2)+\mathcal{T}_2(\n_1,-\n_2)\Bigg]+d(1-d)\mathcal{T}_2(\n_1,0)
     \\\\
     \m_1 &=& \dfrac{(1-d)^2}{2}\Bigg[ \mathcal{T}_2(\n_1,\n_2)-\mathcal{T}_2(\n_1,-\n_2)\Bigg]+d(1-d)\mathcal{T}_2(0,\n_2)
    \end{array}
    \label{eq:equazioni_K2_beta}
\end{equation}
where we used
\begin{equation}
\footnotesize
    \begin{array}{lll}
         \mathcal{T}_2(x,y)=\mathbb{E}_\lambda \tanh\left[\b\left(x+y+\lambda\sqrt{\rho\Big(x^2+y^2\Big)}\right)\right],
         \\\\
         \mathcal{S}_2(x,y)=\mathbb{E}_\lambda \tanh^2\left[\b\left(x+y+\lambda\sqrt{\rho\Big(x^2+y^2\Big)}\right)\right].
    \end{array}
\end{equation}
Solving numerically this set of equations we construct the plots presented in Fig.\ref{fig:phase_diagram_K2}.

\begin{figure}
    \centering
    \includegraphics[width=13.5cm]{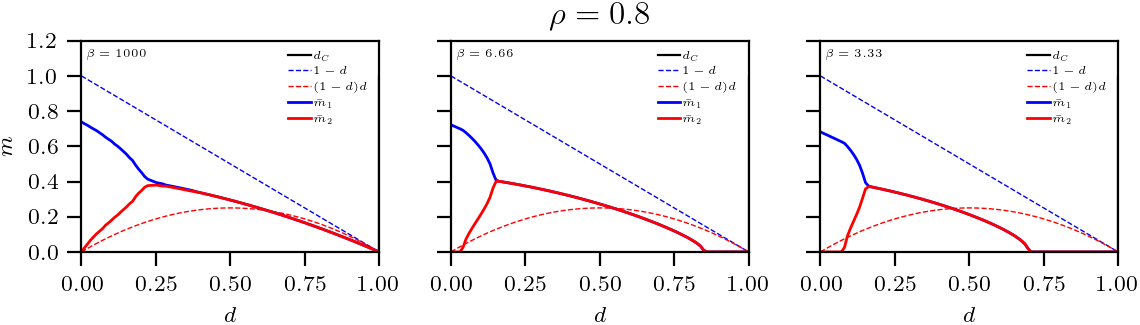}
    \includegraphics[width=13.5cm]{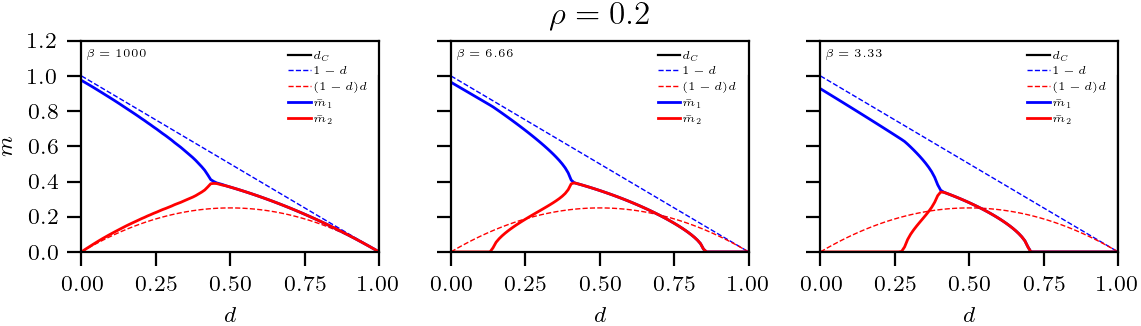}
    \includegraphics[width=13.5cm]{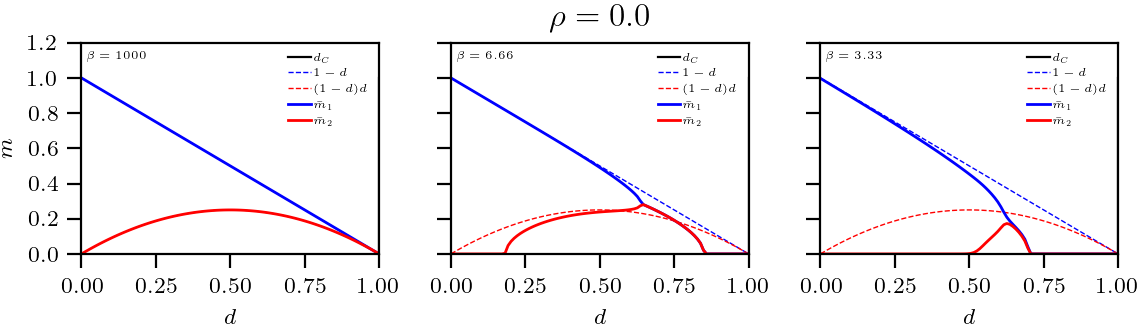}
    \caption{Numerical resolution of the system of equations \eqref{eq:equazioni_K2_beta} for $K=2$: we plot the behaviour of the magnetization $\m$ versus degree of dilution $d$ for fixed $r=0.2$ and different value of $\beta$ (from right to left $\beta=1000,6.66,3.33$) and $\rho$ (from top to bottom $\rho=0.8, 0.2, 0.0$). We stress that for $\rho=0.0$ we recover the standard diluted model presented in Fig.\ref{fig:UnoTaccitua}.}
    \label{fig:phase_diagram_K2}
\end{figure}

\subsection{$K=3$}
Moving on the case of $K=3$ by following the same steps of the previous subsection, we get
\begin{equation}
\footnotesize
    \begin{array}{lll}
    \n_1&=&\dfrac{\m_1}{(1+\rho)}+\dfrac{\b(1-d)\rho\,\n_1}{(1+\rho)}\Bigg\{ 1-d\dfrac{(1-d)}{2}\Bigg[ \mathcal{S}_3(\n_1,\n_2,0)+\mathcal{S}_3(\n_1,0,\n_3)+\mathcal{S}_3(\n_1,-\n_2,0)+\mathcal{S}_3(\n_1,0,-\n_3)\Bigg]
    \\\\
    &&-d^2\mathcal{S}_3(\n_1,0,0)-\dfrac{(1-d)^2}{4}\Bigg[ \mathcal{S}_3(\n_1,\n_2,\n_3)+\mathcal{S}_3(\n_1,\n_2,-\n_3)+\mathcal{S}_3(\n_1,-\n_2,\n_3)+\mathcal{S}_3(\n_1,-\n_2,-\n_3)\Bigg]\Bigg\}\,,
    \\\\
     \m_1 &=& \dfrac{(1-d)^3}{4}\Bigg[ \mathcal{T}_3(\n_1,\n_2,\n_3)+\mathcal{T}_3(\n_1,\n_2,-\n_3)+\mathcal{T}_3(\n_1,-\n_2,\n_3)+\mathcal{T}_3(\n_1,-\n_2,-\n_3)\Bigg]
     \\\\
     &&+d\dfrac{(1-d)^2}{2}\Bigg[ \mathcal{T}_3(\n_1,\n_2,0)+\mathcal{T}_3(\n_1,0,\n_3)+\mathcal{T}_3(\n_1,-\n_2)+\mathcal{T}_3(\n_1,0,-\n_3)\Bigg]+d^2(1-d)\mathcal{T}_3(\n_1,0,0)\,,
    \end{array}
    \label{eq:equazioni_K3_beta}
\end{equation}
where we used
\begin{equation}
\footnotesize
\begin{array}{lll}
\mathcal{T}_3(x,y,z)=\mathbb{E}_\lambda\tanh\left[\b\left(x+y+z+\lambda\sqrt{\rho\Big(x^2+y^2+z^2\Big)}\right)\right]\,,
\\\\
\mathcal{S}_3 (x,y,z)=\mathbb{E}_\lambda\tanh^2\left[\b\left(x+y+z+\lambda\sqrt{\rho\Big(x^2+y^2+z^2\Big)}\right)\right]\,.
\end{array}
\end{equation}
In order to lighten the presentation, we reported only the expression of $\m_1$ and $\n_1$, the related expressions of $\m_2$($\m_3$) and $\n_2$($\n_3$) can be obtained by making the simple substitutions $\m_1\longleftrightarrow \m_2(\m_3)$, and $\n_1\longleftrightarrow \n_2(\n_3)$ in \eqref{eq:equazioni_K3_beta}.
The numerical solution of the previous set of equations is depicted in Fig.\ref{fig:phase_diagram_K3}.

\begin{figure}
    \centering
    \includegraphics[width=15cm]{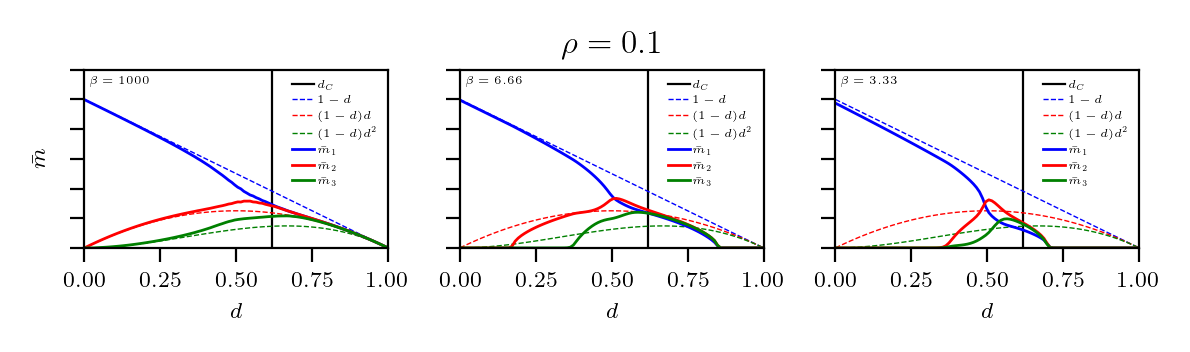}
    \includegraphics[width=15cm]{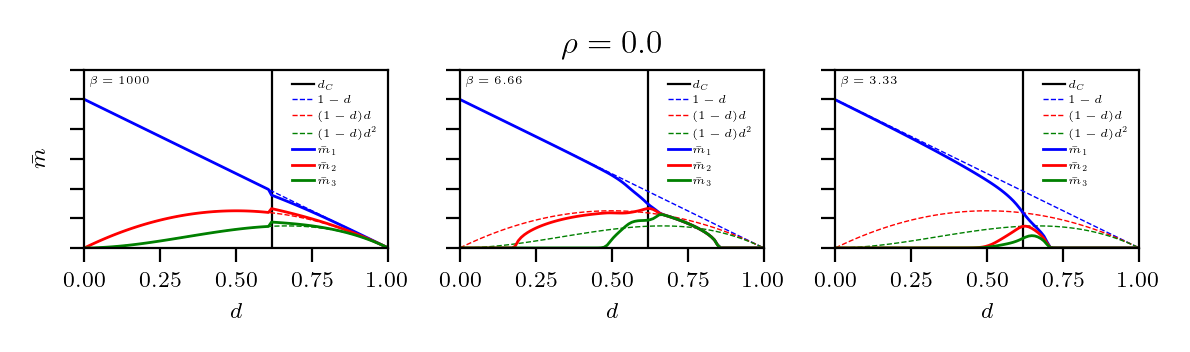}
    \caption{Numerical solution of the system of equations \eqref{eq:equazioni_K3_beta} for $K=3$: we plot the behavior of the magnetization $\m$ versus the degree of dilution $d$ for fixed $r=0.2$ and different value of $\beta$ (from left to right $\beta=1000,6.66,3.33$) and $\rho$ (from top to bottom $\rho=0.8, 0.2, 0.0$).}
    \label{fig:phase_diagram_K3}
\end{figure}

\section{Proofs}\label{app:proof}

\subsection{Proof of Theorem \ref{prop:highnoise}}
\label{proof:highnoise}
In this subsection we show the proof Proposition \ref{prop:highnoise}. 
In order to prove the aforementioned proposition, we put in front of it the following

\begin{lemma} 
The $t$ derivative of interpolating pressure is given by 
\begin{equation}
\begin{array}{lll}
     \dfrac{d\mathcal A^{(sup,unsup)}_{N,K,\beta, d,M,r}}{dt}=&\dfrac{\beta}{2(1-d)}(1+\rho)\SOMMA{\mu=1}{K}\mathbb{E} \,\omega_t[ n^2_\mu ]-\SOMMA{\mu=1}{K}\psi_\mu\mathbb{E} \,\omega_t[ n_\mu].
\end{array}
\end{equation}
\end{lemma}

Since the computation is lengthy but not cumbersome we decided to omit it.

\begin{proposition}
\label{ass:RSassumption}
In the low load regime, in the thermodynamic limit the distribution of the generic order parameter $X$ is centres at its expectation value $\bar X$ with vanishing fluctuations. Thus, being $\Delta X = X-\bar X$,  in the thermodynamic limit, the following relation holds
\begin{equation}
    \mathbb{E} \,\omega_t[ (\Delta X)^2] \xrightarrow[N\to +\infty]{} 0\,.
\end{equation}
\end{proposition}

\begin{remark}
\label{remark:RSassumptionMiriam}
We stress that afterwards we use the relations

\begin{equation}
 \begin{array}{lll}
     \mathbb{E} \,\omega_t[(n_{\mu}-\bar{n}_\mu)^2 ] =\mathbb{E} \,\omega_t[ n_{\mu}^2 ] -2\,\bar{n}_\mu\mathbb{E} \,\omega_t[ n_{\mu} ]+ \bar{n}_\mu^{2}  \,.
\end{array}  
\label{eq:RS_pq_PotenzialiMiriam}
\end{equation}
which are computed with brute force with Newton's Binomial. 

Now, using these relations, if we fix the constants as 
\begin{equation}
    \begin{array}{lll}
        \psi_\mu= \dfrac{\beta}{1-d}(1+\rho)\n_\mu
    \end{array}
\end{equation}
\normalsize
in the thermodynamic limit, due to Proposition \ref{ass:RSassumption}, the expression of derivative w.r.t. $t$ becomes 
\begin{equation}
\begin{array}{lll}
\dfrac{d\mathcal{A}^{(sup,unsup)}_{K,\beta, d,M,r}}{dt}=&-\dfrac{\beta}{2(1-d)}(1+\rho)\SOMMA{\mu=1}{K}\n_\mu^{2}.
\end{array}
\label{eq:dert_TD}
\end{equation}
\normalsize
\end{remark}

\begin{proof}
Let us start from finite size $N$ expression. We apply the Fundamental Theorem of Calculus: 
\small
\begin{equation}
    \mathcal{A}^{(sup,unsup)}_{N, K,\beta, d,M,r}=\mathcal{A}^{(sup,unsup)}_{N, K,\beta, d,M,r}(t=1)=\mathcal{A}^{(sup,unsup)}_{N, K,\beta, d,M,r}(t=0)+\int\limits_0^1\, \partial_s \mathcal{A}^{(sup,unsup)}_{N, K,\beta, d,M,r}(s
)\Big|_{s=t}\,dt.
\label{eq:F_T_CalculusMiriam}
\end{equation}
\normalsize
We have already computed the derivative w.r.t. $t$ in Eq. \eqref{eq:dert_TD}. It only remains to calculate the one-body term:
\small
\begin{equation}
\begin{array}{lll}
      \mathcal{Z}^{(sup,unsup)}_{N, K,\beta, d,M,r}(t=0)=&\SOMMA{\lbrace\boldsymbol{\sigma}\rbrace}{}\exp{\Bigg[\SOMMA{i=1}{N}\left(\SOMMA{\mu=1}{K}\dfrac{\psi_\mu}{2(1+\rho)} \hat{\eta}^\mu+J\xi^\mu\right)\sigma_i\Bigg]}.
\end{array}
\end{equation}
\normalsize
Using the definition of quenched statistical pressure \eqref{hop_GuerraAction} we have 
\small
\begin{equation}
\begin{array}{lll}
      & \mathcal{A}^{(sup,unsup)}_{K,\beta, d,M,r}(J, t=0)=\ln{}\left[{2\cosh\left(\SOMMA{\mu=1}{K}\dfrac{\psi_\mu}{2(1+\rho)} \hat{\eta}^\mu+J\xi^\mu\right)}\right]
      \\\\
      &=\mathbb{E}\left\lbrace\ln{}{2\cosh\left[\dfrac{\beta}{1-d}\SOMMA{\mu=1}{K}\n_\mu\hat{\eta}^\mu+J\xi^\mu\right]}\right\rbrace 
\end{array}
\label{onebody}
\end{equation}
\normalsize
where $\mathbb{E}=\mathbb{E}_\xi \mathbb{E}_{(\eta|\xi)}$. Finally, putting inside \eqref{eq:F_T_CalculusMiriam} \eqref{onebody} and \eqref{eq:dert_TD}, we reach the thesis. 
\end{proof}

\subsection{Proof of Proposition \ref{prop:m_n}}
\label{proof:prop_m_n}

In this subsection we show the proof Proposition \ref{prop:m_n}. 
\begin{proof}
For large data-sets, using the Central Limit Theorem we have 
\begin{equation}
    \hat{\eta}^\mu \sim \xi^\mu\left(1+\sqrt{\rho}\;Z_\mu\right)\,.
    \label{eq:CLT_sup}
\end{equation}
where $Z_\mu$ is a standard Gaussian variable $Z_\mu \sim \mathcal{N}(0,1)$. Replacing Eq. \eqref{eq:CLT_sup} in the self-consistency equation for $\n$, namely Eq. \eqref{eq:n_selfRS_sup}, and applying  Stein's lemma\footnote{This lemma, also known as Wick's theorem, applies to standard Gaussian variables, say $J \sim \mathcal N(0, 1)$, and states that, for a generic
function $f(J)$ for which the two expectations $\mathbb{E}\left( J f(J)\right)$ and $\mathbb{E}\left( \partial_J f(J)\right)$ both exist, then
\begin{align}
    \label{eqn:gaussianrelation2}
    \mathbb{E} \left( J f(J)\right)= \mathbb{E} \left( \frac{\partial f(J)}{\partial J}\right)\,.
\end{align}
} in order to recover the expression for $\m_\mu$, we get the large data-set equation for $\n_\mu$, i.e. Eq. \eqref{eq:n_Mgrande}. 

We will use the relation 
\begin{align}
    \mathbb{E}_{\lambda_\mu}\left[F\left(a+\SOMMA{\mu=1}{K}b_\mu \lambda_\mu\right)\right]=\mathbb{E}_{_Z}\left[F\left(a+Z\sqrt{\SOMMA{\mu=1}{K}b^2_\mu}\right)\right] \, ,
\end{align}
where $\lambda_\mu$ and $Z$ are i.i.d. Gaussian variables. Doing so, we obtain
\begin{equation}
\footnotesize
    g(\beta,\bm \xi, Z, \bar{\bm n})=\b\SOMMA{\nu=1}{K}\n_\nu\xi^\nu +\b\sqrt{\rho\,}\SOMMA{\nu=1}{K}Z_\nu\n_\nu^2\left(\xi^\nu\right)^2 =\b\left(\SOMMA{\nu=1}{K}\n_\nu\xi^\nu +Z\SOMMA{\nu=1}{K}\sqrt{\rho\,\n_\nu^2\left(\xi^\nu\right)^2} \right) \ ,
\end{equation}
thus we reach the thesis.
\end{proof}

\begin{corollary}\label{Corollario3}
The self consistency equations in the large data-set assumption and  null-temperature limit are 
\begin{equation}
    \begin{array}{lll}
        \m_\mu &=&   \mathbb{E}_{\xi} \left\{\mathrm{erf}\left[\left(\SOMMA{\nu=1}{K}\m_\nu\xi^\nu\right)\left(2\rho\SOMMA{\nu=1}{K}\m_\nu^2\left(\xi^\nu\right)^2\right)^{-1/2} \right] \xi^\mu\right\}.
    \end{array}
\end{equation}
\end{corollary}
\begin{proof}
In order to lighten the notation we rename
\begin{equation}
    C = \tanh^2\left[g(\beta, \bm\xi, Z, \bar{\bm n})\right]\,.
\end{equation}
We start by assuming finite the limit
\begin{equation}
    \lim_{\b\to\infty}\b(1-C) = D \in \mathbb{R}
\end{equation}
and we stress that as $\b \rightarrow \infty$ we have $C \rightarrow 1$.
As a consequence, the following reparametrization is found to be useful,
\begin{equation}
    C=1-\dfrac{\delta C}{\b}\;\;\;\;\mathrm{as}\;\;\;\;\b\to \infty.
\end{equation}
Therefore, as $\b\to \infty$, it yields
\begin{equation}
    \begin{array}{lll}
        \n_\mu = \dfrac{\m_\mu}{1+\rho-\rho\,\delta C\,(1-d)}
        \\\\
         \m_\mu = \mathbb{E}_\xi\mathbb{E}_Z\left[\mathrm{sign}\left(\SOMMA{\nu=1}{K}\n_\nu\xi^\nu +Z\SOMMA{\nu=1}{K}\sqrt{\rho\,\n_\nu^2\left(\xi^\nu\right)^2} \right) \xi^\mu\right]\, ;
    \end{array}
    \label{eq:new_beta_inf}
\end{equation}
to reach this result, we have also used the relation
\begin{equation}
    \mathbb{E}_z\mathrm{sign}[A+Bz]=\mathrm{erf}\left[\dfrac{A}{\sqrt{2}B}\right] \ ,
\end{equation}
where $z$ is a Gaussian variable $\mathcal{N}(0,1)$ and the truncated expression $\n_\mu= \m_\mu/(1+\rho)$ for the first equation in \eqref{eq:new_beta_inf}.
\end{proof}

\section*{Acknowledgments}
This research has been supported by Ministero degli Affari Esteri e della Cooperazione Internazionale (MAECI) via the BULBUL grant (Italy-Israel), CUP Project n. F85F21006230001, and has received financial support from the Simons Foundation (grant No. 454949, G. Parisi) and ICSC – Italian Research Center on High Performance Computing, Big Data and Quantum Computing, funded by European Union – NextGenerationEU.
\newline
Further, this work has been partly supported by The Alan Turing Institute through the Theory and Methods Challenge Fortnights event {\em Physics-informed machine learning}, which took place on 16-27 January 2023 at The Alan Turing Institute headquarters.
\newline
E.A. acknowledges financial support from Sapienza University of Rome (RM120172B8066CB0).
\newline
E.A., A.A. and A.B. acknowledges GNFM-INdAM (Gruppo Nazionale per la Fisica Mamematica, Istituto Nazionale d’Alta Matematica), A.A. further acknowledges UniSalento  for financial support via PhD-AI and AB further acknowledges the PRIN-2022 project {\em Statistical Mechanics of Learning Machines: from algorithmic and information-theoretical limits to new biologically inspired
paradigms}.

\end{document}